%% file: CMCS18_DescriptionPreservingPredicateLiftings.tex
\documentclass[envcountsame,envcountsect,envcountreset,orivec,runningheads,a4paper,final]{llncs}
\usepackage[T1]{fontenc}
\usepackage[utf8]{inputenc}

\usepackage{ifthen}
\newboolean{full}
\setboolean{full}{true}

\usepackage{ifdraft}

\usepackage{chngcntr} 
\input{defslncs}
\usepackage{lmodern}
\usepackage{microtype}

\usepackage{amssymb,amsmath,amsbsy} 
\usepackage{mathtools}
\usepackage{stmaryrd}
\SetSymbolFont{stmry}{bold}{U}{stmry}{m}{n}
\usepackage{dsfont}
\usepackage{nicefrac}

\usepackage{xspace}
\usepackage{enumerate}

\usepackage{tikz}
\usetikzlibrary{cd,automata,positioning} 

\usepackage{graphicx}

\newcounter{blubber}

\newenvironment{sparenumerate}
{\begin{list}
  {(\arabic{blubber})}
  {\usecounter{blubber}
   \setlength{\leftmargin}{0pt}
    \setlength{\parsep}{0pt}
    \setlength{\itemindent}{4ex}
    \setlength{\itemsep}{2pt}
  }
}
{\end{list}}

\iftrue
\ifdraft{
\usepackage[marginparwidth=2cm,nohead,nofoot,
            headsep = 1cm,
            footskip = 1cm,
            text={12.2cm,19.3cm},
            paperwidth=16.2cm,
            paperheight=23.3cm,
            marginparsep=1mm,
            marginparwidth=1.8cm,
            footskip=1cm,
            centering
            ]{geometry}
}{}
\fi

\usepackage[draft]{fixme}

\usepackage[draft=false,pdftex,hypertexnames=false,colorlinks]{hyperref} 
\usepackage{hypcap} 

\hypersetup{hidelinks}

\usepackage{seqsplit}
\usepackage{xstring}

\makeatletter
\ifdraft{%
\usepackage[notcite,notref]{showkeys}%

\renewcommand*\showkeyslabelformat[1]{%
\@ifundefined{hideNextShowKeysLabel}{%
\noexpandarg%
\StrSubstitute{#1}{ }{\textvisiblespace}[\TEMP]%
\parbox[t]{\marginparwidth}{\raggedright\normalfont\small\ttfamily\(\{\){\color{red!50!black}\expandafter\seqsplit\expandafter{\TEMP}}\(\}\)}%
}{}
}

\hypersetup{hidelinks}
}{}
\makeatother

\counterwithin*{equation}{section}
\counterwithin*{equation}{subsection}

\let\doendproof\endproof
\renewcommand\endproof{~\hfill\qed\doendproof}

\fxuselayouts{inline,nomargin,marginclue}
\fxusetheme{color}
\FXRegisterAuthor{ls}{als}{LS}
\FXRegisterAuthor{sm}{asm}{SM}
\FXRegisterAuthor{ud}{aud}{UD}
\FXRegisterAuthor{tw}{atw}{TW}
\FXRegisterAuthor{rf}{arf}{Referee}

\tikzcdset{
  arrow style=tikz,
  diagrams={>=Straight Barb,column sep=huge,row sep=huge}
}

\newcommand{\myparagraph}[1]{\medskip\par\noindent\textbf{\textsf{#1}}\hspace{6pt}}

\newcommand{\val}{\ensuremath{\kappa}\xspace}
\newcommand{\sem}[2][]{\ensuremath{\llbracket #2 \rrbracket_{#1}}\xspace}
\newcommand{\semv}[3]{\ensuremath{\llbracket #3 \rrbracket_{#1}^{#2}}\xspace}

\newcommand{\FA}{\ensuremath{\mathfrak A}\xspace}
\newcommand{\FB}{\ensuremath{\mathfrak B}\xspace}

\newcommand{\pow}{\ensuremath{\mathcal P}\xspace}
\newcommand{\copow}{\ensuremath{\mathcal Q}\xspace}
\newcommand{\fpow}{\ensuremath{\mathcal P_\omega}\xspace}
\newcommand{\dist}{\ensuremath{\mathcal D}\xspace}
\newcommand{\mon}{\ensuremath{\mathcal M}\xspace}
\newcommand{\fmon}{\ensuremath{\mathcal M_\omega}\xspace}
\newcommand{\B}{\ensuremath{\mathcal B}\xspace}
\newcommand{\Reals}{\ensuremath{\mathbb R}\xspace}

\newcommand{\fixvars}{\ensuremath{\mathsf V}\xspace}
\newcommand{\pliftset}{\ensuremath{\Lambda}\xspace}
\newcommand{\plift}{\ensuremath{\lambda}\xspace}
\newcommand{\modalset}{\ensuremath{\mathcal{L}}\xspace}
\newcommand{\modalsym}{\ensuremath{\textit{L}}\xspace}
\newcommand{\fpre}{\ensuremath{\alpha}\xspace}
\newcommand{\cfpre}{\ensuremath{\tau}\xspace}

\DeclareMathOperator{\ar}{ar}

\DeclareMathOperator{\inl}{inl}
\DeclareMathOperator{\inr}{inr}

\newcommand\irel[1]{\ensuremath{#1^\circ}\xspace}

\newcommand{\expr}{\ensuremath{\mathcal{E}}\xspace}
\newcommand{\gcexpr}{\ensuremath{\mathcal{E}_0}\xspace}
\newcommand{\gcemor}{\ensuremath{\varepsilon}}

\newcommand{\cat}[1]{\ensuremath{\mathbf{#1}}\xspace}
\newcommand{\epi}{\ensuremath{\twoheadrightarrow}}
\newcommand{\incl}{\ensuremath{\hookrightarrow}}

\numberwithin{equation}{section}
\def\op{\mathsf{op}}
\def\Set{\cat{Set}}
\newcommand{\takeout}[1]{\empty}
\newcommand{\modab}[3]{\langle #1, a.#2, b.#3\rangle}

\spnewtheorem{construction}[definition]{Construction}{\bfseries}{\rm}
\spnewtheorem{examples}[definition]{Examples}{\bfseries}{\rm}

\begin{document}

\mainmatter 

\title{Predicate Liftings and Functor Presentations in Coalgebraic Expression Languages\thanks{Work forms part of the DFG project COAX (MI 717/5-1 / SCHR 1118/11-1)}}
%
\titlerunning{Predicate Liftings and Functor Presentations}

\author{
  Ulrich Dorsch\orcidID{0000-0001-8678-7876}
  \and Stefan Milius\orcidID{0000-0002-2021-1644}
  \and \\Lutz Schr\"oder\orcidID{0000-0002-3146-5906}
  \and Thorsten Wi{\ss}mann\orcidID{0000-0001-8993-6486}
}
%
\authorrunning{U.~Dorsch, S.~Milius, L.~Schr\"oder, T.~Wi{\ss}mann}

\institute{Friedrich-Alexander-Universit\"at Erlangen-N\"urnberg, Erlangen, Germany 
\email{\{ulrich.dorsch,stefan.milius,lutz.schroeder,thorsten.wissmann\}@fau.de }
}

\maketitle


\begin{abstract}
  We introduce a generic expression language describing behaviours of
  finite coalgebras over sets; besides relational systems, this
  covers, e.g., weighted, probabilistic, and neighbourhood-based
  system types. We prove a generic Kleene-type theorem establishing a
  correspondence between our expressions and finite systems. Our
  expression language is similar to one introduced in previous work by
  Myers but has a semantics defined in terms of a particular form of
  predicate liftings as used in coalgebraic modal logic; in fact, our
  expressions can be regarded as a particular type of modal fixed
  point formulas. The predicate liftings in question are required to
  satisfy a natural preservation property; we show that this property
  holds in particular for the Moss liftings introduced by Marti and
  Venema in work on lax extensions.
\end{abstract}

\section{Introduction}

Expression languages that support the syntactic description of system
behaviour are one of the classical topics in computer science. The
prototypic example are regular expressions; further examples include
Kleene algebra with tests~\cite{Kozen97} and expression languages for
labelled transition systems~\cite{AcetoHennessy92}.

There has been recent interest in phrasing such expression languages
generically, obtaining their syntax and semantics as well as
meta-theoretic results including Kleene theorems by instantiation of a
parametrized framework. This is achieved by abstracting the type of
systems as coalgebras for a given type functor. This line of work
originates with expression languages for a specific class of functors
that essentially covers relational systems, so-called \emph{Kripke
  polynomial functors}~\cite{SilvaBR10}, and was subsequently extended
to cover also weighted systems~\cite{SilvaBBR11}. A generic expression
language for arbitrary finitary functors can be based on algebraic
functor presentations~\cite{M13}. Here, we introduce a similar and, as
it will turn out, in fact largely equivalent generic expression
language for finitary functors, which we base on coalgebraic
modalities in predicate lifting style, following the paradigm of
coalgebraic logic~\cite{CirsteaEA11}; on predicate liftings, we impose
strong conditions, notably including preservation of singletons. Marti
and Venema~\cite{MartiV15} have shown that for functors admitting a
\emph{lax extension} (in particular for functors that admit a
separating set of monotone predicate liftings), one can convert
operations from the functor presentation into predicate liftings, the
so-called \emph{Moss liftings}.  We show that the Moss liftings
preserve singletons; the converse does not hold in general, i.e.\ not
all singleton-preserving predicate liftings are Moss liftings under a
given lax extension.

  We thus arrive at a generic expression language that covers, e.g.,
  various flavours of relational, weighted, and probabilistic systems,
  as well as monotone neighbourhood systems as in the semantics of
  game logic~\cite{Parikh83} and concurrent dynamic
  logic~\cite{Peleg87}. We prove a Kleene theorem stating that every
  expression denotes the behavioural equivalence class of some state
  in a finite system, and that conversely every such behavioural
  equivalence class is denoted by some expression. 

  We make no claim to novelty for the design of a generic expression
  language as such, and in fact the expression language developed by
  Myers in his PhD dissertation~\cite{M13} appears to be even more
  general. In particular, unlike Myers' language our expression
  language is currently restricted to describing behavioural
  equivalence classes in set-based coalgebras, and does not yet
  support algebraic operations (e.g.~a join semilattice structure as
  in Silva et al.'s language for Kripke-polynomial
  functors~\cite{SilvaBR10} or in fact in standard regular
  expressions). The main point we are making is, in fact, a different
  one: we show that
  \begin{quote}
    \emph{coalgebraic expression languages embed into coalgebraic
      logic},
  \end{quote}
  specifically into (the conjunctive fragment of) the coalgebraic
  $\mu$-calculus~\cite{Cirstea11}, extending the classical result that
  every bisimilarity class of states in finite labelled transition
  systems is expressible by a \emph{characteristic formula} in the
  $\mu$-calculus~\cite{GrafSifakis86,GodskesenEA87,SteffenIngolfsdottir94,AcetoEA07}. This
  result provides a direct link between descriptions of processes and
  their property-oriented specification; as indicated above, the key
  to lifting it to a coalgebraic level of generality are
  singleton-preserving predicate liftings.

  \myparagraph{Related Work} As mentioned above, we owe much to work
  by Marti and Venema on Moss liftings~\cite{MartiV15}, and moreover
  we use a notion of
  \emph{$\Lambda$-bisimulation}~\cite{GorinSchroder13} that turns out
  to be an instance of their definition of bisimulation via lax
  extensions. Besides the mentioned work on generic expression
  languages for Kripke polynomial~\cite{SilvaBR10},
  weighted~\cite{SilvaBBR11}, and finitary~\cite{M13} functors, there
  is work on expression languages for \emph{reactive
    $T$-automata}~\cite{GoncharovEA14}, which introduce an orthogonal
  dimension of genericity: The coalgebra functor as such remains fixed
  but the computational capacities of the automaton model at hand are
  encapsulated as a computational
  monad~\cite{Moggi91}. Venema~\cite{Venema06} proves that for
  weak-pullback preserving fuctors, every bisimilarity class of finite
  coalgebras is expressible in coalgebraic fixpoint logic over Moss'
  $\nabla$ modality.

\section{Preliminaries}

\noindent In the standard paradigm of universal coalgebra, types of
state-based systems are encapsulated as endofunctors. We recall
details on presentations of set functors and on their
property-oriented description via predicate-lifting based coalgebraic
modalities.

\subsubsection*{Functor Presentations} describe set functors by
signatures of operations and a certain restricted form of equations,
so-called \emph{flat} equations, alternatively by a suitable natural
surjection. A \emph{signature} is a sequence
$\Sigma = (\Sigma_n)_{n \in \omega}$ of sets. Elements of $\Sigma_n$
are regarded as $n$-ary operation symbols (we write
$\cfpre/n \in \Sigma$ for $\cfpre \in \Sigma_n$). Every signature $\Sigma$
determines the corresponding polynomial endofunctor $T_\Sigma$ on
$\Set$, which maps a set $X$ to the set
\[
  T_\Sigma X = \coprod\limits_{n\in \omega} \Sigma_n \times X^n
\]
and similarly on maps.

\begin{defn}\label{def:fnpre}
  A \emph{presentation} of a functor $T: \Set \to \Set$ is a pair
  $(\Sigma, \fpre)$ consisting of a signature $\Sigma$ and a natural
  transformation $\fpre : T_\Sigma \epi T$ with surjective components
  $\fpre_X$.  In the following, we abuse notation and denote, for every
  $\cfpre/n \in \Sigma$, the corresponding coproduct component of
  $\fpre: T_\Sigma \epi T$ again by $\cfpre: (-)^n \to T$, and refer to
  it as an \emph{operation} of $T$.
\end{defn}

\noindent Most of our results concern finitary set functors. Recall
that a functor is \emph{finitary} if it preserves filtered
colimits. Over $\Set$, we have the following equivalent
characterizations:
\begin{thm}[Ad\'amek and Trnkova~\cite{AdamekT90}]\label{thm:at}
  Let $T: \Set \to \Set$ be a functor. Then the following are
  equivalent:
  \begin{enumerate}
  \item $T$ is finitary;
  \item $T$ is \emph{bounded}, i.e.~for every element $x \in TX$ there
    exists a finite subset $m: Y \incl X$ and an element $y \in TY$
    such that $x = Tm(y)$;
  \item $T$ has a presentation. 
  \end{enumerate} 
\end{thm}
Indeed, for the equivalence of (1) and (3) note that every polynomial
functor~$T_\Sigma$ is finitary, and finitary functors are closed under
taking quotient functors. Conversely, given a finitary functor $T$,
let $\Sigma_n = Tn$ and define $\fpre_X: T_\Sigma X \to TX$ by
$\fpre_X(\cfpre, t) = Tt(\cfpre)$, where $t \in X^n$ is considered as a
function $n \to X$. It is easy to show that this yields a natural
transformation with surjective components.

\begin{rem}
  As indicated above, the natural surjection $\alpha$ in a functor
  presentation $(\Sigma,\alpha)$ can be replaced with a set of flat
  equations over~$\Sigma$, where an equation is called \emph{flat} if
  both sides consist of an operation symbol applied to
  variables~\cite{AdamekT90}. Incidentally, this (standard) term
  should not be confused with the same term introduced in the context
  of our expression language in Section~\ref{sec:expr}.
\end{rem}
\begin{expl}\label{ex:pres}
  \begin{sparenumerate}
  \item Let $A$ be an input alphabet. The functor $TX = 2 \times X^A$,
    whose coalgebras are deterministic automata, is polynomial, and
    finitary if~$A$ is finite. Thus, $T$ has a presentation
    $(\Sigma,\alpha)$ by a signature $\Sigma$ with two $|A|$-ary
    operations and no equations, i.e.~$\alpha$ is the natural
    isomorphism $T_\Sigma \cong 2 \times (-)^A$.
  \item For a commutative monoid $(M,+,0_M)$ the monoid-valued functor
    $M^{(-)}:\Set\to\Set$ is defined by
    \[
      M^{(X)} = \{\mu: X\to M\mid \mu(x) = 0_M
      \text{ for all but finitely many }x\in X\}
    \]
    and by $M^{(h)}(\mu) = y \mapsto \sum_{h(x) = y} \mu(x)$ on maps
    $h: X\to Y$.  We view elements of $M^{(X)}$ as finitely supported
    additive measures on~$X$, and in particular write
    $\mu(A)=\sum_{x\in A}\mu(x)$ for $A\subseteq X$; in this view,
    maps $M^{(h)}$ just take image measures. For a set $G\subseteq M$
    of generators (i.e.\@~there exists a surjective monoid morphism
    $G^*\epi M$), $M^{(-)}$ is represented by
      \[
      \alpha_X: \coprod_{n\in \omega} G^n\times X^n
    \twoheadrightarrow M^{(X)}, \qquad
    \alpha_X (\cfpre, t) = M^{(t)}(\cfpre),
  \]
  where $\cfpre \in G^n$ is considered as an element of $M^{(n)}$.
\item The finite powerset functor $\fpow$ (with $\fpow(X)$ being the
  set of finite subsets of~$X$) is the monoid-valued functor for the
  monoid $(\{0,1\},\vee,0)$. Since this is generated by $G=\{1\}$, we
  have one $n$-ary operation symbol for each $n\in \omega$:
    \[
    \fpre_X : \coprod_{n\in \omega} X^n \epi \pow_\omega X, \quad \quad
    \fpre_X (x_1, \ldots, x_n) = \{ x_1, \ldots x_n \};
    \]
    e.g.~$\alpha$ identifies the tuples  $(x_1, x_1, x_2)$ and $(x_1, x_2)$.
  \item For the monoid $\mathds{N}$ of natural numbers with addition,
    one obtains the bag functor $\B$ as
    $\mathds{N}^{(-)}$. Concretely, $\B$ maps a set $X$ to the set
    $\B X$ of bags (i.e.\@~finite multisets) on $X$. Since
    $(\mathds{N},+,0)$ is generated by $G=\{1\}$, we have the same
    signature as for $\fpow$, namely one $n$-ary operation symbol per
    $n\in \omega$; of course, the presentation $\alpha$ now identifies
    fewer tuples, e.g.\@ distinguishes $(x_1,x_2,x_1)$ and $(x_1,x_2)$.
  \item The \emph{finite distribution functor}~$\dist$ is a subfunctor of
    the monoid-valued functor $\Reals_{\ge 0}^{(-)}$ for the additive
    monoid of the non-negative reals, given by
    $\dist X=\{\mu\in\Reals_{\ge 0}^{(-)}\mid \sum_{x\in
      X}\mu(x)=1\}$. Note that elements of $\dist X$ can be represented
    as formal convex combinations $\sum_{i=1}^n p_ix_i$,
    $p_i \in \Reals_{\ge 0}, x_i \in X$ for $i = 1, \ldots, n$, with
    $p_1 + \cdots + p_n = 1$. Taking $\Reals_{\ge 0}$ itself as the set
    of generators and restricting to $\dist$, we obtain a presentation
    $(\Sigma, \alpha)$ with an $n$-ary operation symbol for each
    $n$-tuple $(p_1,\dots,p_n)\in\Reals_{\ge 0}^n$ such that
    $p_1+\cdots+p_n=1$, and $\alpha_X$ maps
    $((p_1, \ldots, p_n), (x_1, \ldots, x_n))$ to the formal convex
    combination $\sum_{i=1}^n p_ix_i$.
  \item The \emph{finitary monotone neighbourhood functor} $\fmon$,
    i.e.\ the finitary part of the standard monotone neighbourhood
    functor $\mon$, can be described as follows. To begin, $\mon$ is
    the subfunctor of the double contravariant powerset functor
    $\copow\copow^{\op}$ given on objects by
    \begin{equation*}
      \mon X=\{\FA\subseteq\copow(X)\mid 
      \FA\text{ upwards closed under $\subseteq$}\}.
    \end{equation*}
    We can then describe $\fmon X$ as consisting of all
    $\FA\in\mon X$ having finitely many minimal elements, all of
    them finite, such that every element of \FA is above a
    minimal one. We have the following presentation of~$\fmon$: For
    every choice of numbers $n\ge 0$, $k_1,\dots,k_n\ge 0$, we have a
    $\sum_{i=1}^n k_i$-ary operation mapping
    $(x_{ij})_{i=1,\dots,n;j=1,\dots,k_i}$ to the upwards closure of
    the set system
    \begin{equation*}
      \{\{x_{i1},\dots,x_{ik_i}\}\mid i=1,\dots,n\}.
    \end{equation*}
  \end{sparenumerate}
\end{expl}

\subsubsection*{Coalgebraic Logic} Since coalgebras serve as generic
models of reactive systems, it is natural to specify properties of
coalgebras in terms of suitable modalities. The semantics of
coalgebraic modalities can be defined using \emph{predicate
  liftings}~\cite{pattinson2003coalgebraic,Schroder08}, which specify
how a predicate on a base set~$X$ induces a predicate on the set~$TX$
where~$T$ is the coalgebraic type functor:

\begin{defn}\label{def:plift}
  For $n\in\omega$ an $n$-ary \emph{predicate lifting} for a functor $T: \Set
  \rightarrow \Set$ is a natural transformation
  \[
    \plift : \copow{}^n \rightarrow \copow{T^{\op}}
  \]
  where $\copow{} : \mathbf{Set}^{\op} \rightarrow \mathbf{Set}$ is
  the contravariant powerset functor, with $\copow f$ taking
  preimages, i.e.
  \begin{equation*}
    \copow f (A) = f^{-1}[A].
  \end{equation*}
  We write $\plift/n$ to indicate that $\plift$
  has arity~$n$.  A predicate lifting $\plift$ is \emph{monotone} if
  it preserves set inclusion in every argument.  A set $\pliftset$ of
  predicate liftings is
  \emph{separating}~\cite{Pattinson04,Schroder08} if every $t\in TX$
  is uniquely determined by the set
  \[
    \mathcal{T}_\pliftset(t) = \{ (\plift, A_1, \ldots, A_n) \mid \plift/n\in\pliftset,
      A_i\in \copow X \text{ and } t\in\plift_X(A_1, \ldots,A_n) \}.
  \]
\end{defn}
\begin{expl}\label{expl:pl}
  The basic example is the interpretation of the standard box modality
  $\Box$ over the covariant powerset functor $\pow$ (with $\pow f$
  taking direct images), given by the monotone unary predicate
  lifting~$\plift$ defined by
  \begin{equation*}
    \plift_X(A)=\{B\in\pow(X)\mid B\subseteq A\}.
  \end{equation*}
  For a further monotone example, we interpret the box modality over
  the monotone neighbourhood functor~$\mon$ (Example~\ref{ex:pres})
  by the monotone unary predicate lifting
  \begin{equation*}
    \plift_X(A)=\{\FA\in\mon X\mid A\in\FA\}.
  \end{equation*}
  It is easy to see that in both these examples, the predicate lifting
  for $\Box$ alone is separating.
\end{expl}
\noindent Predicate-lifting-based modalities can be embedded into
\emph{coalgebraic logics} of varying degrees of expressiveness. Our
expression language introduced in Section~\ref{sec:expr} will live
inside the \emph{coalgebraic $\mu$-calculus}~\cite{Cirstea11}, more
precisely its \emph{conjunctive fragment}~\cite{GorinSchroder14}. We
defer details to Section~\ref{sec:expr}.

\section{Singleton-Preserving Predicate Liftings}
\noindent Our generic expression language will depend on a specific
type of predicate liftings, as well as on a strengthening of separation:
\begin{defn}
  An $n$-ary predicate lifting $\plift$ \emph{preserves singletons} if
  \[
     |\plift_X(\{x_1\},\dots,\{x_n\})| = 1 
  \]
  for all $x_1,\dots,x_n\in X$. Moreover, a set \pliftset of predicate
  liftings is \emph{strongly expressive} if for every $t \in TX$ there
  exist $\plift/n\in\pliftset$ and $x_1, \ldots, x_n \in X$ such that
  \begin{equation*}
    \{t\} = \plift_X(\{x_1\}, \ldots, \{x_n\}).
  \end{equation*}
\end{defn}
\noindent Singleton preservation will serve to ensure that
expressions of our language denote unique behaviours, while strong
expressivity will guarantee that all (finite) behaviours are
expressible. The following is immediate:
\begin{lem}\label{lem:sep}
  Every strongly expressive set of predicate liftings is separating.
\end{lem}
\begin{expl}\label{expl:sp-pl}
  The predicate liftings in Example~\ref{expl:pl} both fail to
  preserve singletons. Our main source of singleton-preserving
  predicate liftings are Moss liftings as introduced in general terms
  in the next section. For the finite powerset functor~$\fpow$
  consider the predicate liftings $\plift^n/n$ given by
  \begin{equation}\label{eq:moss-pfin}\textstyle
    \begin{array}{r@{}l}
      \plift^n_X(A_1,\ldots,A_n)=\{
      B\in\fpow X\mid\, & B \subseteq\bigcup_{i=1}^nA_i \text{ and }\\
      & B\cap A_i\neq\emptyset\text{ for $i=1,\ldots,n$}\}
    \end{array}
  \end{equation}
  (which can be seen as arising from the above lifting for $\Box$ by
  Boolean combination). Then
  $\plift^n_X(\{x_1\},\dots,\{x_n\})=\{\{x_1,\dots,x_n\}\}$ for
  $x_1,\dots,x_n\in X$, which shows that the $\plift^n$ preserve
  singletons and that the set $\{\plift^n\mid n\in\omega\}$ is strongly
  expressive.
\end{expl}

\begin{rem}
  Singleton-preserving predicate liftings should not be confused with
  Kurz and Leal's \emph{singleton
    liftings}~\cite{leal2008predicate,kurz2009equational}. The
  definition of the latter is based on the one-to-one correspondence
  between subsets of $T( 2^n )$ and $n$-ary predicate liftings for
  $T$~\cite{Schroder08}, which maps an $n$-ary predicate lifting
  $\plift$ to
  $\plift_{2^n}( \pi_1^{-1}(\{\top\}),\ldots,\pi_n^{-1}(\{\top\}) )
  \subseteq T( 2^n )$,
  and $C\subseteq T(2^n)$ to the lifting $\plift$ defined by
  $\plift_{X}( A_1,\ldots,A_n ) = \{ t\in TX\mid
  T\langle\chi_{A_1},\ldots,\chi_{A_n}\rangle(t) \in C\}$,
  where $\pi_i: 2^n \to 2$ is the $i$-th projection and
  $\chi_A: X \to 2$ denotes the characteristic function of
  $A\subseteq X$.  An $n$-ary predicate lifting is a \emph{singleton
    lifting} if it corresponds to a singleton subset of $T(2^n )$.

  It is then indeed immediate that every \emph{unary}
  singleton-preserving predicate lifting $\plift$ is a singleton
  lifting, since the above correspondence maps $\plift$ to the
  singleton $\plift_2( \{ \top \} )$. The following examples show that
  this implication breaks down at higher arities, and that the
  converse also fails in general.
\end{rem}
\begin{expl}
  \begin{sparenumerate}
  \item The unary singleton lifting for $\pow$ corresponding to
    $\{\{\bot\}\}\subseteq\pow 2$ fails to preserve singletons. Of
    course, this lifting fails to be monotone.
  \item Binary monotone singleton liftings need not preserve
    singletons. E.g.~for the distribution functor $\dist$, the monotone
    singleton lifting $\plift$ corresponding to
    $\{ 1\cdot ( \top,\top) \}\subseteq\dist( 2^2 )$ is given by
    $\plift( A, B ) = \{ \mu \mid \mu(A) = \mu(B) = 1 \} $, so
    $\plift(\{ x \},\{ y \}) = \varnothing$ for $x\neq y$. We leave it
    as an open question whether unary monotone singleton liftings
    preserve singletons.
    \item The binary singleton-preserving predicate lifting 
      \[ \plift( A,B ) = \{ \mu
      \mid \mu(A) \geq 1/2, \mu(B)\geq1/2,\mu(A\cup B) = 1
      \}
      \]
      for the distribution functor \dist (see Example~\ref{ex:Moss}
      for details) is not a singleton lifting, as it corresponds to the
      following infinite subset of $\dist( 2^2 )$:
      \[
        \{ \mu \mid \mu( 2\times \{\top\} )\geq
          1/2, \mu(\{\top\}\times 2)\geq1/2,\mu(
          2\times\{\top\}\cup\{\top\}\times 2 ) = 1\}.
      \]
  \end{sparenumerate}
\end{expl}

\noindent It is not hard to see that we can recover operations for a
functor from \emph{monotone} singleton preserving predicate
liftings; in detail:
\begin{lem}\label{lem:pres}
  Let $T:\Set\to\Set$. Then the following hold.
  \begin{enumerate}
  \item\label{item:op-from-pl} For each monotone
    singleton-preserving predicate lifting $\plift/n$,
    \begin{equation}\label{eq:op-from-pl}
      \{\cfpre_{\plift, X}( x_1, \ldots, x_n )\} :=
      \plift_X( \{ x_1 \},\ldots,\{ x_n \}
    )
  \end{equation}
  defines a natural transformation $\cfpre_\plift : (-)^n \to T$.
\item\label{item:pres-from-pls} If $\pliftset$ is a strongly expressive
  set of monotone singleton-preserving predicate liftings, then
  taking operation symbols $\cfpre_\plift$ for each
  $\plift\in\pliftset$, with associated interpretation as
  per~\eqref{eq:op-from-pl}, yields a functor presentation of~$T$.
  \end{enumerate}
\end{lem}

\begin{expl}\label{expl:moss}
  The singleton-preserving predicate liftings $\plift^n$ from
  Example~\ref{expl:pl} induce, according to the above construction,
  the operations $X^n\to\fpow(X)$,
  $(x_1,\dots,x_n)\mapsto\{x_1,\dots,x_n\}$.
\end{expl}
\noindent The other direction, generating predicate liftings from
functor presentations, is more involved, and treated next.


\section{Moss Liftings}

Marti and Venema~\cite{MartiV15} introduce \emph{Moss liftings},
predicate liftings that are constructed from functor presentations
with the help of a generalized form of the nabla operator, extending
an earlier construction for weak-pullback preserving functors by Kurz
and Leal~\cite{kurz2009equational}.  Recall that for a
weak-pullback-preserving functor~$T$, Moss'~\cite{Moss99} classical
nabla operator $\nabla : T\copow \Rightarrow \copow T^\op$ is the
natural transformation defined by
\[
  \nabla (\Phi) = \{ t\in TX \mid ( t, \Phi ) \in \overline
  T(\in_X)\}.
\]
Here, $\mathord{\in_X} \subseteq X\times \copow X$ is the element-of
relation for $X$, and $\overline T$ is the \emph{Barr extension}
of~$T$, viz.~the functor $\overline T$ on the category of sets and
relations defined on a relation $R \subseteq X \times Y$ by
$\overline T R = \{ (T\pi_1 (r), T\pi_2 (r)) \mid r\in TR
\}$, where $\pi_1: R \to X$ and $\pi_2: R \to Y$ are the
projection maps (cf.~\cite{Moss99}). Barr~\cite{Barr70} (see
also Trnkov\'a~\cite{Trnkova80}) proved that $\overline T$ is a functor if and
only if $T$ preserves weak pullbacks.

Further recall that the \emph{converse} of a relation
$R \subseteq X \times Y$ is the relation
$\irel R = \{(y,x) \mid x \mathbin R y \}\subseteq Y \times X$. We
denote the composite of two relations $R \subseteq X \times Y$ and
$S \subseteq Y \times Z$ diagrammatically by
$R;S \subseteq X \times Z$. Also, for $A\subseteq X$ we denote by
$R[A]\subseteq Y$ the relational image
$R[A]=\{y\mid\exists x\in A.\,xRy\}$. The construction
$T\mapsto\overline T$ is generalized and abstracted in the notions of
\emph{relation lifting} and, more specifically, \emph{lax extension}
of a functor, as recalled next.

\begin{defn}[Relation lifting, lax extension~\cite{MartiV15}]
  \label{def:relation-lifting}
  A \emph{relation lifting} $L$ for a functor $T$ is an assignment
  mapping every relation $R \subseteq X\times Y$ to a relation
  $LR \subseteq TX\times TY$ such that converses are preserved:
  $L(\irel S) = \irel{(LS)}$.  A relation lifting~$L$ is a \emph{lax
    extension} if for all relations $R, R' \subseteq X\times Z$,
  $S\subseteq Z\times Y$ and functions $f:X \rightarrow Z$
  (identified with their graph relation) the following
  hold:
  \begin{gather*}
  R' \subseteq R \Rightarrow LR' \subseteq LR,\\
  LR;LS \subseteq L( R;S ),\\
  Tf \subseteq Lf.
  \end{gather*}
  A lax extension $L$ \emph{preserves diagonals} if for all sets~$X$
  \[
    L\Delta_X \subseteq \Delta_{TX}.
  \]
\end{defn}

\begin{prop}[Properties of Lax Extensions \cite{MartiV15}]
  \label{prop:laxprop}
  Let $L$ be a lax extension for a functor $T$. Then for all functions
  $f:X\rightarrow Z$, $g:Y\rightarrow Z$ and relations $R \subseteq X\times Z$,
  $S \subseteq Z\times Y$, 
  \begin{itemize}
    \item[i)] $\Delta_{TX} \subseteq L\Delta_X$,
    \item[ii)] $Tf;LS = L( f;S )$ and $LR;\irel{( Tg )} = L( R;\irel g
    )$,
  \end{itemize}
  and if $L$ preserves diagonals, then
  \begin{itemize}
    \item[iii)] $\Delta_{TX} = L\Delta_X$ and $Tf = Lf$,
    \item[iv)] $Tf;\irel{( Tg )} = L( f;\irel g )$.
  \end{itemize}
\end{prop}
\noindent One use of relation liftings is to determine coalgebraic notions of
bisimulation:
\begin{defn}[$L$-Bisimulation \cite{MartiV15}]
  Let $L$ be a relation lifting for a functor $T:\Set \to \Set$, and let
  $( X,\xi )$, $( Y,\zeta )$ be $T$-coalgebras. A relation
  $S \subseteq X\times Y$ is an \emph{$L$-simulation} if for all
  $x\in X$ and $y\in Y$, 
  \begin{equation*}
    x\mathbin{S}y\quad\text{implies}\quad\xi(x) \mathbin{LS} \zeta(y).
  \end{equation*}
  An \emph{$L$-bisimulation} is a relation $S$ such that $S$ and
  $\irel S$ are $L$-simulations.  Two states are \emph{$L$-bisimilar}
  if there exists an $L$-bisimulation relating them.
\end{defn}
Marti and Venema~\cite[Theorem~11]{MartiV15} show that if~$L$ is a lax
extension that preserves diagonals, then $L$-bisimilarity coincides
with behavioural equivalence.

\begin{assumption}
  From now on we fix a finitary endofunctor $T:\Set \to \Set$ having a
  diagonal-preserving lax extension~$L$ and a presentation
  $(\Sigma, \fpre)$ of~$T$.
\end{assumption}
\noindent Another key feature of lax extensions is that they induce
canonical modalities, generalizing Moss' coalgebraic
logic~\cite{Moss99}:
\begin{defn}[Lax Nabla \cite{MartiV15}]
  The \emph{lax nabla} of $L$ is the  family of functions
  \begin{align*}
    \nabla^L_X : T\copow X & \rightarrow \copow T^\op X\\
    \Phi\quad &\mapsto \{ t\in TX \mid ( t, \Phi ) \in L(\in_X) \},
  \end{align*}
  where ${\in_X} \subseteq X\times \copow X$ is the element-of relation for $X$.
\end{defn}
As shown by Marti and Venema~\cite{MartiV15}, the lax nabla is in fact
a natural transformation $\nabla^L: T\copow \Rightarrow \copow T^\op$,
and coincides with Moss' classical $\nabla$ for $L$ being the Barr
extension of $T$ (and $T$ preserving weak pullbacks). In combination
with a functor presentation, the lax nabla gives rise to a family of
predicate liftings:

\begin{defn}[Moss Liftings \cite{MartiV15}]\label{def:Mosslift}
  Every operation symbol $\cfpre/n \in \Sigma$ yields a predicate
  lifting $\plift$ defined by
    \begin{equation*}
      \plift =
      (\copow^{n} \overset{\cfpre \copow}
      {\Longrightarrow} 
      T\copow
      \overset{\nabla^L}
      {\Longrightarrow}
      \copow T^\op),
    \end{equation*}
    that is,
    \begin{equation*}
      \plift_X( X_1, \ldots X_n) =
      \{ t\in TX \mid ( t, \cfpre_{\copow X}( X_1,\ldots,X_n ) )\in
      L(\in_X) \}.
    \end{equation*}
    These predicate liftings are called the \emph{Moss liftings} of
    $T$.
\end{defn}
\begin{expl}\label{ex:Moss}
  Some standard functor presentations are converted into Moss liftings
  as follows.
\begin{sparenumerate}
\item For the deterministic automata functor $TX = 2 \times X^A$
  consider the Barr extension $L = \overline T$. Then elements of
  $T\copow X$ are pairs $(b, (Y_a)_{a \in A})$, where each~$Y_a$ is a
  subset of $X$, and
  \[
    \nabla_X(b, (Y_a)_{a \in A}) 
    = 
    \{ (b, (x_a)_{a \in A}) \mid \forall a \in A : x_a \in Y_a\}
    \qquad \text{for $b = 0, 1$}.
  \] 
  The two Moss liftings
  $\plift^0, \plift^1: \copow^A \to \copow (2 \times (-)^A)$
  corresponding to the two $|A|$-ary operation symbols from the
  presentation in \autoref{ex:pres}.1 are thus defined (slightly
  abusing notation) by
  \[
    \plift^i((Y_a)_{a \in A}) = \{(i, (x_a)_{a \in A}) \mid \forall a
    \in A : x_a \in Y_a\} \qquad \text{for $i = 0,1$}.
  \]

\item As indicated in \autoref{ex:pres}, the finite powerset functor
  $\fpow$ has operations $\cfpre^n/n$ given by
  $\cfpre^n(x_1,\dots,x_n)=\{x_1,\dots,x_n\}$. The Moss lifting
  $\plift^n$ associated to $\cfpre^n$ when using the Barr extension is
  exactly the one given by~\eqref{eq:moss-pfin} above.
\item Recall from \autoref{ex:pres} that the operations of the finite
  distribution functor~$\dist$ take formal convex combinations. Via the
  Barr extension, such an operation, determined by coefficients
  $p_1,\dots,p_n$ such that $\sum p_i=1$, induces the predicate
  lifting $\plift$ given by $\plift_X(A_1,\dots,A_n)$ consisting of
  all $\mu\in\dist X$ such that there exists a distribution on $\in_X$ (a
  subset of $X\times\copow(X)$) whose marginal distributions are $\mu$
  (on $X$) and the distribution $\nu$ on $\copow(X)$ given by
  $\nu(\{A_i\})=p_i$, respectively. In fact, however, this description
  can be substantially simplified; e.g.\ one readily checks that in
  the case $n=2$, we actually have
  \[
    \plift(A_1,A_2)
    =
    \{\mu\in\dist(X)\mid \mu(A_1)\ge p_1,\mu(A_2)\ge p_2,\mu(A_1\cup A_2)=1\}.
  \]
  (The generalization to higher arities is via what is nowadays known
  as the \emph{splitting lemma}~\cite[Theorem~11]{Strassen65}.)
  \item For the finitary monotone neighbourhood functor~$\fmon$
  (Example~\ref{ex:pres}), we obtain Moss liftings as follows. Marti
  and Venema~\cite{MartiV15} define a diagonal-preserving lax
  extension~$L$ for $\mon$ (which, then, restricts to $\fmon$) by
  means of nested Egli-Milner liftings. An explicit description
  of~$L$ is
  \begin{equation*}
    LR = \{(\FA,\FB)\in\mon X\times\mon Y\mid
    \forall A\in\FA.\,R[A]\in\FB,
    \forall B\in\FB.\,\irel R[B]\in\FA\}
  \end{equation*}
  for $R\subseteq X\times Y$. In particular, for $\FA\in\mon X$ and
  $\Phi\in\mon\copow X\subseteq\copow\copow\copow X$, we have
  \begin{align*}
    \FA\in\nabla^L_X(\Phi)\quad\text{iff}\quad \FA\mathbin{L(\in)} \Phi\quad\text{iff}\quad&
    \textstyle\forall \beta\in\Phi.\,\bigcup\beta\in\FA\text{ and }\\
    &\forall A\in\FA.\,\{B\in\copow X\mid B\cap A\neq\emptyset\}\in\Phi.
  \end{align*}
  Combining $\nabla^L$ with the presentation of~$\fmon$
  (Example~\ref{ex:pres}) produces, for each choice of numbers
  $n\ge 0$ and $k_1,\dots,k_n\ge 0$, a $\sum_{i=1}^n k_i$-ary Moss
  lifting~$\plift$ given by
  \begin{align*}
    \plift((A_{ij})_{i=1,\dots,n;j=1,\dots,k_i})=\{\FA\in\fmon X\mid\,
    &\textstyle\forall i.\,\bigcup_j A_{ij}\in\FA\text{ and }\\
    &\forall B\in\FA.\,\exists i.\,\forall j.\,B\cap A_{ij}\neq\emptyset\}.
  \end{align*}
  Since $\fmon$ preserves finite sets and the box modality $\Box$ as
  described in Example~\ref{expl:pl} is separating, it is clear that
  the Moss liftings are expressible using~$\Box$ and Boolean
  operators. Concretely, this works as follows. For readability, we
  denote the predicate lifting interpreting $\Box$ by $\Box$ as well,
  similarly for the dual modality~$\Diamond$, so that
  $\Diamond_X(A):=\mon X\setminus\Box_X(X\setminus A)=\{\FA\in\mon X\mid
  \forall B\in\FA.\,B\cap A\neq\emptyset\}$.
  Then the Moss lifting~$\plift$ as described above can be written as
  \begin{equation*}\textstyle
    \plift((A_{ij}))=\bigcap_i \Box_X(\bigcup_j A_{ij})\cap
    \bigcap_\pi\Diamond_X(\bigcup_i A_{i\pi(i)})
  \end{equation*}
  where $\pi$ ranges over all selection functions assigning to each
  $i\in\{1,\dots,n\}$ an index $\pi(i)\in\{1,\dots,k_i\}$.

\end{sparenumerate}
\end{expl}

\noindent Moss liftings are always
monotone~\cite[Proposition~24]{MartiV15}. We show that they also
preserve singletons:

\begin{prop}
  \label{prop:laxdescrpres}
  Moss liftings preserve singletons. More specifically, let
  $\plift$ be the Moss lifting induced by $\cfpre/n \in \Sigma$. Then
  for all $x_1, \ldots, x_n \in X$,
  \[
    \plift_X( \{ x_1 \},\dots,\{ x_n \} )
    = \{ \cfpre_X( x_1,\dots,x_n) \}.
  \]
\end{prop}
\noindent Marti and Venema already establish that the Moss liftings
are separating~\cite[Proposition~25]{MartiV15}; we show that they are
even strongly expressive:
\begin{prop}\label{prop:moss-strong-expr}
  The set $\pliftset$ of all Moss liftings of $T$ is strongly
  expressive.
\end{prop}

\begin{remark}
  Incidentally, this also means that for finitary functors the
  existence of a separating set of monotone predicate liftings is
  equivalent to the existence of a strongly expressive set of monotone
  singleton-preserving predicate liftings. The right-to-left
  implication is trivial; the converse follows from \autoref{prop:laxdescrpres},
  \autoref{prop:moss-strong-expr},
  and the fact that for finitary functors the existence of a separating set of monotone
  predicate liftings is equivalent to the existence of a lax extension
  \cite{MartiV15}.
\end{remark}
We have thus seen that given a fixed diagonal-preserving lax extension, from every natural
transformation $\cfpre: (-)^n \to T$ we obtain the corresponding Moss
lifting $\plift^\cfpre/n$, which is a monotone singleton-preserving
predicate lifting. Conversely, every monotone singleton-preserving
predicate lifting $\plift$ yields a natural transformation
$\cfpre^\plift: (-)^n \to T$
(\autoref{lem:pres}.\ref{item:op-from-pl}). From
\autoref{prop:laxdescrpres}, it is immediate that for
$\cfpre: (-)^n \to T$ we have
\begin{equation*}
\cfpre = \cfpre^{(\plift^\cfpre)}.
\end{equation*}
In particular, taking Moss liftings is an injection from functor
operations to monotone singleton-preserving predicate
liftings. Conversely, however, $\plift = \plift^{(\cfpre^\plift)}$ need
not hold in general -- recall that the construction of Moss liftings
depends on the choice of a diagonal-preserving lax extension, and a
functor may have more than one such extension. We report an example
due to Paul Levy:
\begin{example}
  Let $M$ be the monoid of non-negative reals. This monoid in fact
  forms a division semiring in the expected sense
  (e.g.~\cite{Weinert84}), i.e.\ it is a semiring, and its non-zero
  elements form a multiplicative group. We note that every division
  semiring is \emph{refinable} in the sense of Gumm and
  Schröder~\cite{GummSchroder01}, i.e.~$n$ specified row sums
  $b_1,\dots,b_n$ and $k$ specified column sums $c_1,\dots,c_k$ that
  induce the same total sum $d=\sum b_i=\sum c_j$ can always be
  realized by some $n\times k$-matrix $(a_{ij})$ -- in fact, one can
  just put $a_{ij}=b_ic_j/d$.
  Now let $b\in(0,1)$ be a transcendental number, and let
  $N\subseteq M$ be generated by~$b$ in~$M$ as a division
  semiring. Concretely, elements of~$N$ have the form $f(b)/g(b)$
  where $f(X)$ and $g(X)\neq 0$ are polynomials with non-negative rational
  coefficients. In particular, $1-b\notin N$: If we could write $1-b$
  in the prescribed form $f(b)/g(b)$, then by transcendentality of
  $b$, $f(X)/g(X)=1-X$, in contradiction to the leading coefficients
  of $f$ and~$g$ being positive.
  
  Both $M$ and $N$ are \emph{positive} ($x+y=0$ implies $x=y=0$)
  and refinable, so that the monoid-valued functors $F=M^{(-)}$ and
  $G=N^{(-)}$ both preserve weak pullbacks~\cite{GummSchroder01}. As
  recalled above, it follows that in both cases, the Barr extension is
  functorial, in particular is a diagonal-preserving lax
  extension. Now diagonal-preserving lax extensions are easily seen to
  be inherited by subfunctors, so that the Barr extension $\overline F$
  induces a diagonal-preserving lax extension~$L$ of~$G$. This
  extension differs from the Barr extension~$\overline G$; we immediately
  cast the counterexample in the form that interests us here:

  Let $X=\{u,v\}$. Representing elements of $GX$ as formal linear
  combinations, we have a binary functor operation $\cfpre(x,y)=x+by$
  for~$G$. We write $\plift^1$ and~$\plift^2$ for the Moss liftings
  induced from~$\cfpre$ via $\overline G$ and via $L$, respectively (by the
  above, both $\plift^1$ and $\plift^2$ induce~$\cfpre$). Then
  $u+bv\in\plift^1(\{u,v\},\{u\})$ but
  $u+bv\notin\plift^2(\{u,v\},\{u\})$: For the former, we have a
  unique witnessing element of $F{\in_X}$, namely
  $(1-b)(u,\{u,v\})+b(v,\{u,v\})+b(u,\{u\})$; but in $G\,\!{\in_X}$,
  there is no witnessing element since $1-b\notin N$.

\end{example}
\noindent Summing up, even for weak-pullback preserving functors,
singleton-preserving monotone predicate liftings are not in general
uniquely determined by the functor operation they induce. In the above
example, both singleton predicate liftings inducing the given functor
operation arise as Moss liftings, via different diagonal-preserving
lax extensions; we currently do not know whether every
singeleton-preserving monotone predicate lifting is a Moss lifting for
some diagonal-preserving lax extension.
\begin{remark}\label{rem:moss-unary}
  It is fairly easy to see that for monotone singleton-preserving
  \emph{unary} predicate liftings~$\plift$, we do have
  $\plift = \plift^{(\cfpre^\plift)}$. 
\end{remark}


\section{Generic Expressions}\label{sec:expr}

We proceed to define, given a set of monotone and singleton-preserving
predicate liftings for a functor $T$, syntactic expressions describing
the behaviour of states of $T$-coalgebras. Our main result is a
Kleene-type theorem stating that for every state of a $T$-coalgebra
there exists an equivalent expression, and conversely, every
expression describes the behaviour of some state of a finite
$T$-coalgebra. As indicated above, our expression language is a small
fragment of the coalgebraic $\mu$-calculus, essentially restricted to
modalities and greatest fixed points $\nu z.\,\phi$.

\begin{defn}[Expressions]\label{def:generic-expressions}
  We fix a set $\fixvars$ of \emph{fixed point variables} and a
  set~\modalset of \emph{modalities} equipped with an arity function
  $\ar:\modalset \rightarrow \omega$; we write $L/n \in \modalset$ if
  $L\in\modalset$ and $\ar(L) = n$. The set~\expr of
  \emph{expressions}~$\phi,\dots$ is then defined by the grammar
  \[
    \phi ::= z \mid \nu z.\,\phi \mid \modalsym (\phi_1,\dots\phi_n)
    \qquad (z\in\fixvars,L/n\in\modalset).
  \]
  An expression is \emph{closed} if all its fixed point variables
  are bound by a fixed point operator. An expression is \emph{guarded}
  if all its fixed point variables are separated from their binding fixed
  point operator by at least one modality. We write \gcexpr for the
  set of closed and guarded expressions. We have the usual notion of
  \emph{$\alpha$-equivalence} of expressions modulo renaming of bound
  variables. An occurrence of a fixed point operator in an expression
  is \emph{top-level} if it is not in scope of a modality.
\end{defn}
\noindent We next define the semantics of expressions, which agrees
with their interpretation as formulas in coalgebraic logic. We fix the
requisite data:
\begin{assumption}
  For the rest of the paper, we fix a set \modalset of modalities and an assignment of a
  singleton-preserving monotone $n$-ary predicate lifting $\sem{L}$
  for~$T$ to each $L/n\in\modalset$ such that the set
  $\pliftset:=\{\sem{L}\mid L\in\modalset\}$ is strongly expressive.
\end{assumption}
\noindent By the results of the previous section, these assumptions
imply that $T$ has a presentation and is thus finitary (\autoref{thm:at}).
\begin{defn}[Semantics]\label{def:coalglogsem}
  Given a $T$-co\-al\-ge\-bra $C=(X, \xi)$ and a valuation
  $\val : \fixvars \rightarrow \copow X$, the semantics
  $\semv{C}{\val}{\phi}\subseteq X$ of expressions $\phi\in\expr$ is
  given by
  \begin{align*}
    \semv{C}{\val}{z} & = \val( z ) \\
    \semv{C}{\val}{\modalsym(\phi_1,\dots \phi_n)} & = \xi^{-1}[ \sem[X]{L}
      (\semv{C}{\val}{\phi_1},\dots \semv{C}{\val}{\phi_n})] \\
    \semv{C}{\val}{\nu z. \phi} & = \nu Y .       \semv{C}{\val[ z \mapsto Y ]}{\phi}
  \end{align*}
  where as usual, we use $\nu$ to denote greatest fixed points of monotone
  maps.  When~$\phi$ is closed, we simply write $\sem[C]{\phi}$ in
  lieu of $\semv{C}{\val}{\phi}$, and we drop the subscript~$C$
  whenever~$C$ is clear from the context.
\end{defn}
\noindent Note that since the predicate liftings $\sem{L}$ are
monotone and $\xi^{-1}$ is a monotone map, the requisite greatest
fixed points exist by the Knaster-Tarski fixed point theorem. Moreover, the assumption that the predicate liftings are singleton-preserving will ensure that every expression describes exactly one behavioural equivalence class (see \autoref{thm:generic-expression-describes-exactly-one-bisimulation-class}).  

By dint of the fact that our expression language is contained in the
coalgebraic $\mu$-calculus, the following is an immediate consequence
of the fact that the latter is invariant under behavioural equivalence
(e.g.~\cite{SchroderVenema18}):
\begin{lem}[Invariance under behavioural
  equivalence]\label{lem:invariance}
  For every closed expression~$\phi$ and coalgebras $C=(X,\xi)$,
  $D=(Y,\zeta)$, if states $x\in X$ and $y\in Y$ are behaviourally
  equivalent, then $x\in\sem[C]{\phi}$ iff $y\in\sem[D]{\phi}$.
\end{lem}
\begin{lem}\label{lem:fp-substitution}
  For all expressions $\phi \in \expr$, $\sem{\nu z.\phi} = \sem{\phi[ \nu z.\phi/z ]}$.
\end{lem}

\begin{expl}\label{ex:expr}
\begin{sparenumerate}
\item For the deterministic automaton functor $TX = 2 \times X^A$ with
  $A = \{a,b\}$, we let $\modalset$ be the set of two binary
  modalities $\modab 0 {(-)} {(-)}$ and $\modab 1 {(-)} {(-)}$
  (corresponding to the two Moss liftings of~\autoref{ex:Moss}.1). We
  interpret expressions in the final $T$-coalgebra $\nu T$ carried by
  all formal languages over $A$. Here are a few closed and guarded
  expressions and their semantics in $\nu T$ (as usual $|w|_b$ denotes the
  number of $b$'s in $w$):
  \begin{align*}
    \sem{\nu v. \modab 0 v v} & = \{\emptyset\} \\
    \sem{\nu z. \modab 1 z z} & = \{A^*\} \\
    \sem{\nu x.  \modab 1 x {\nu y. \modab 0 y x}} 
    & 
    = \{\{ w \in A^* \mid |w|_b\ \text{even} \}\}
  \end{align*}
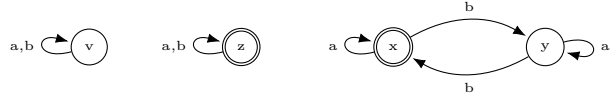
  \[
    \begin{tikzpicture}[>=Latex,shorten >=1pt,on grid,node distance=2cm,every state/.style={minimum
      size=10pt},auto,every node/.style={font=\tiny}]
      \node[state] (q0) {v};
      \path[->] (q0) edge [loop left] node {a,b} ();

      \node[state,accepting] (q1) [right=of q0] {z};
      \path[->] (q1) edge [loop left] node {a,b} ();

      \node[state,accepting] (q2) [right=of q1] {x};
      \node[state] (q3) [right=of q2] {y};
      \path[->] (q2) edge [loop left]  node {a} ()
                     edge [bend left]  node {b} (q3)
                (q3) edge [loop right] node {a} ()
                     edge [bend left] node {b} (q2);
    \end{tikzpicture}
  \]
  Note that the semantics of each of these expressions is a singleton
  (up to behavioural equivalence); in fact, for an arbitrary
  $T$-coalgebra $X$, the semantics of the above expressions is the set
  of states accepting the language in the singleton on the
  right. In~\autoref{lem:gfp-is-singleton} further below we prove that
  this holds in general.
  
\item Consider $T = \pow_\omega(A \times -)$ where $A$ is a finite set of
    labels. A presentation of $T$ is given by the signature containing
    for each $n$-tuple $\vec a = (a_1, \ldots, a_n) \in A^n$ one $n$-ary
    operation symbol, and the corresponding natural transformation
    $\cfpre^{\vec a}: (-)^n \to T$ is defined by
    \[
      \cfpre^{\vec a}_X: (x_1, \ldots, x_n) \mapsto \{(a_1,x_1), \ldots,
      (a_n,x_n)\}. 
    \]
    The corresponding Moss lifting is $\plift^{\vec a}/n$ given by
    \begin{multline*}\textstyle
      \plift^{\vec a}_X(Y_1, \ldots, Y_n)
      =
      \{  Z\in\fpow (A \times X) \mid Z
           \subseteq\bigcup_{i=1}^n(\{a_i\}
           \times Y_i) \\
         \text{and }Z \cap \{a_i\} \times Y_i\neq\emptyset\text{ for
                                      $i=1,\ldots,n$}\}
    \end{multline*}
    (cf.~\eqref{eq:moss-pfin}). Now put
    $\modalset = \{[\vec a]/n \mid \vec a \in A^n, n \in \omega\}$ and
    interpret each $[\vec a]$ by $\plift^{\vec a}$.  For example, for
    $A = \{a,b\}$ the expression $\nu x. [a]([a,b,a](x,[()],[()]))$,
    where $[()]$ is the unique nullary modality in $\modalset$,
    describes the left-hand state in the following labelled transition system
    
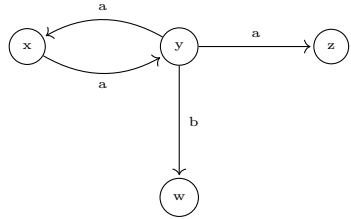
\[
    \begin{tikzpicture}[shorten >=1pt,on grid,node distance=2cm,every state/.style={minimum
      size=10pt},auto,every node/.style={font=\tiny}]
      \node[state] (q0) {x};
      \node[state] (q1) [right=of q0] {y};
      \node[state] (q2) [right=of q1] {z};
      \node[state] (q3) [below=of q1] {w};

      \path[->] (q0) edge [bend right] node[below] {a} (q1)
                (q1) edge [bend right] node[above] {a} (q0)
                     edge node {a} (q2)
                     edge node {b} (q3);
    \end{tikzpicture}
    \]
  \item For $T = \dist$ we have the presentation with an $n$-ary
    operation $\cfpre^{\vec p}$ for every $\vec p = (p_1, \ldots, p_n)$
    with $\sum_{i=1}^n p_i = 1$ and corresponding Moss liftings as
    described in~\autoref{ex:Moss}.3. For each such $\vec p$, we
    introduce a modality $[\vec p]/n\in \modalset$, and interpret it
    as $\plift^{\vec p}$. Now consider the Markov chain (i.e.\
    $\dist$-coalgebra)
    
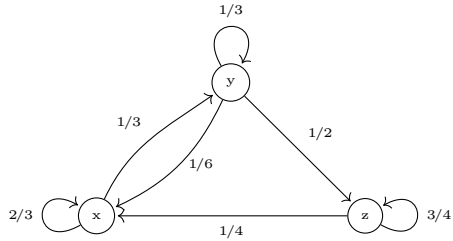
\[
    \begin{tikzpicture}[shorten >=1pt,on grid,node distance=2.5cm,every state/.style={minimum
      size=10pt},auto,every node/.style={font=\tiny}]
      \node[state] (q0) {x};
      \node[state] (q1) [above right=of q0] {y};
      \node[state] (q2) [below right=of q1] {z};

      \path[->] (q0) edge [out=210,in=150,loop]  node {$\nicefrac{2}{3}$} ()
                     edge [out=65,in=215]  node {$\nicefrac{1}{3}$} (q1)
                (q1) edge [out=120,in=60,loop] node {$\nicefrac{1}{3}$} ()
                     edge [out=245,in=25]  node[right] {$\nicefrac{1}{6}$} (q0)
                     edge              node {$\nicefrac{1}{2}$} (q2)
                (q2) edge [out=-30,in=30,loop] node[right] {$\nicefrac{3}{4}$} ()
                     edge              node {$\nicefrac{1}{4}$} (q0);
      \end{tikzpicture}
    \]
    The behaviour of the left-hand state is described by the expression
    \[
      \nu x. [\nicefrac{2}{3}, \nicefrac{1}{3}](x, \nu y.
      [\nicefrac{1}{6},\nicefrac{1}{3},\nicefrac{1}{2}](x,y, \nu z.
      [\nicefrac{1}{4},\nicefrac{3}{4}](x,z))).
    \]
\end{sparenumerate}
\end{expl}
\begin{remark}
  The syntax of our expressions is determined purely by the finitary
  coalgebraic type functor, more precisely, by a given strongly
  expressive set $\pliftset$ of monotone singleton-preserving
  predicate lifting. In contrast, existing expression calculi such as
  standard regular expressions for deterministic automata or the
  coalgebraic expression calculi in~\cite{SilvaBR10,SilvaBBR11} use
  extra operations (e.g.~expressing union or concatenation of
  languages). These operations are not dictated by the setting,
  viz.~an endofunctor on $\Set$. Rob Myers' PhD thesis~\cite{M13}
  explains nicely how such extra operations are obtained naturally in
  an expression calculus when one works over an algebraic category
  (such as the one of join-semilattices or vector spaces over the
  reals, i.e.~algebras for the monad $\Reals^{(-)}$). We leave the
  extension of our expression language to this more general setting
  for future work.
\end{remark}
\noindent Our Kleene theorem requires a number of technical lemmas:

\begin{lem}\label{lem:predicate-liftings-preserve-equivalence-classes}
  Let $\plift/n$ and $\plift'/n'$ be monotone singleton-preserving
  predicate liftings for~$T$. Let $S$ be an equivalence relation on a
  set $X$, let $A_1,\dots,A_{n}$, be
  $S$-equivalence classes or empty, and let $B_1,\dots,B_{n'}$ be $S$-closed
  subsets of $X$. Then the following holds.
  \begin{enumerate}[(1)]
   \item \(
    \plift_{X} (A_1,\dots A_n) \subseteq \plift_{X}'
      (B_1,\dots B_{n'})  \text{ or }
    \plift_{X} (A_1,\dots A_n) \cap \plift_{X}' (B_1,\dots B_{n'}) =
      \varnothing.
   \)
   \item If the $B_1,\ldots,B_{n'}$ are even $S$-equivalence classes or empty, then
  \begin{align*}
    \plift_{X} (A_1,\dots A_n) = \plift_{X}'
      (B_1,\dots B_{n'})  \text{ or }
    \plift_{X} (A_1,\dots A_n) \cap \plift_{X}' (B_1,\dots B_{n'}) =
      \varnothing.
  \end{align*}

   \end{enumerate}
\end{lem}
\begin{proof}[Sketch]
  Apply naturality to the quotient map $q: X \epi X/S$.
\end{proof}
\noindent In the proof of \autoref{lem:gfp-is-singleton} further
below, we will make use of a $\pliftset$-bisimulation. We briefly
recall the essentials of this notion~\cite{GorinSchroder13}:
\sloppypar%
\begin{defn}[$\pliftset$-Simulation]
  Given a pair of $T$-coalgebras $( X,\xi )$ and $( Y,\zeta )$, a
  \emph{$\pliftset$-simulation} is a relation $S\subseteq X\times Y$ such
  that for all predicate liftings $\plift \in \pliftset$ and
  $X_i\subseteq X$, $x\mathbin{S}y$ implies
  \[
    \xi( x ) \in \plift_X( X_1,\dots, X_n ) \Rightarrow \zeta( y
    )\in\plift_Y( S[ X_1 ],\dots, S[ X_n ] ).
  \]
  A \emph{$\pliftset$-bisimulation} is a $\pliftset$-simulation $S$ such that
  $\irel{S}$ is also a $\pliftset$-simulation. Elements $(x,y) \in X\times Y$
  are \emph{$\pliftset$-bisimilar} if there is a $\pliftset$-bisimulation
  relating $x$ and~$y$.
\end{defn}
\begin{thm}
  \label{thm:Lambda}
  \pliftset-bisimilarity concides with behavioural equivalence.
\end{thm}
\begin{rem}
  In fact, for Theoren~\ref{thm:Lambda} it is sufficient that
  \pliftset is separating and the predicate liftings in \pliftset are
  monotone. It turns out that Theorem~\ref{thm:Lambda} is actually a
  special case of~\cite[Theorem~11]{MartiV15}, applied to the case
  where the lax extension is induced by a separating set of monotone
  predicate liftings.
\end{rem}

\begin{lem}\label{lem:gfp-is-singleton}
  Let $(X,\xi)$ be a $T$-coalgebra, let $\plift_i/k\in\pliftset$,
  $i = 1, \ldots, k$, and let $(A_1,\dots, A_k)$ be the greatest fixed
  point of the map $h:(\copow X)^k \rightarrow (\copow X)^k$ defined
  by
  \begin{equation}\label{eq:flat-gfp}
    \begin{pmatrix}
      X_1 \\
      \vdots \\
      X_k
    \end{pmatrix} \mapsto 
    \begin{pmatrix}
      \xi^{-1}[\plift_{1,X}(X_1,\dots,X_k)] \\
      \vdots \\
      \xi^{-1}[\plift_{k,X}(X_1,\dots, X_k)]
    \end{pmatrix}
  \end{equation}
  Then for each $i$, all elements of~$A_i$ are behaviourally
  equivalent, and for all~$i$,~$j$, either $A_i\cap A_j=\varnothing$
  or $A_i=A_j$.
\end{lem}
\noindent (In the above lemma, we restrict to all $\plift_i$ having
full arity~$k$ and using their arguments in the given order only in
the interest of readability; this is w.l.o.g.\ since we can just
reorder arguments and add dummy arguments.)
\begin{proof}[Sketch]
  Let $S \subseteq X \times X$ be the relation
  \[
    S = \{ ( x_1,x_2 )\mid \exists A_i\,.\,x_1\in A_i\wedge x_2\in A_i
    \}\cup\Delta_X.
  \]
  Using \autoref{lem:predicate-liftings-preserve-equivalence-classes}
  one shows first that~$S$ is an equivalence relation, which already
  takes care of the second part of the claim, and then that~$S$ is a
  $\Lambda$-bisimulation. The first claim of the lemma then follows by
  Theorem~\ref{thm:Lambda}.
\end{proof}

\noindent The final ingredient of our Kleene-type correspondence is
the following adaptation of Bekič's bisection lemma \cite{Bekic84}:
\begin{lem}\label{lem:bekic}
  For complete lattices $(X,\le)$, $(Y,\le)$ and for every pair of monotone maps
  $f: X\times Y \rightarrow X$ and $g: X\times Y \rightarrow Y$, we have
  \[
    \nu ( x, y ) . ( f( x,y ),g( x,y ) ) = (x_0, y_0)
    \quad\text{ with }\quad
    \begin{array}{rl}
    x_0 &= \nu x.f( x, \nu y . g( x, y ) ) \\
    y_0 &= \nu y. g( x_0, y ).
    \end{array}
  \]
\end{lem}
Although in \cite{Bekic84} this lemma only covers least fixed points
in a slightly different setting, the proof is the same.
\iffull For completeness, we provide it in the appendix.\fi

Using \autoref{lem:bekic} we can transform every expression
$\phi\in\gcexpr$ into a system of flat equations
$(z_1 = \phi_1,\dots, z_k = \phi_k)$ for some $k\in\omega$, i.e.\@
equations without nested modalities or fixed point operators: This is
done by first ensuring that every fixed point operator uses a
different fixed point variable and then binding every modality that is
not nested directly under a fixed point operator with a new fixed
point operator using a fresh variable. Thus we can rewrite every expression
$\phi\in\gcexpr$ in the form
\[
  \phi \equiv \nu z_1 .\, \modalsym_1 ( z_1, \nu z_2 .
  \modalsym_2( \dots ),\dots, \nu z_k .
  \modalsym_k( \dots ))
\]
for some modalities $L_i\in\modalset$, $i = 1, \ldots, k$.  If we
now inductively apply \autoref{lem:bekic} and, for readability,
additionally normalize every modality to have as many arguments as
there are different fixed point variables in such an expression,
introducing dummy arguments where necessary, then we can write
$\phi$ as a system
\begin{equation}\label{eq:system}
\begin{aligned}
  z_1 &= \modalsym_1( z_1, z_2, \dots, z_k )\\
  z_2 &= \modalsym_2( z_1, z_2, \dots, z_k )\\
  \vdots &\\
  z_k &= \modalsym_n( z_1, z_2, \dots, z_k )
\end{aligned}
\end{equation}
of flat equations. Given any coalgebra $C = (X, \xi)$, the above
system induces an obvious map of the form~\eqref{eq:flat-gfp} (replacing
$z_i$ by $X_i$ and $L_i$ by $\plift_i= \sem{L_i}$), and the first
components of its greatest fixed point is the semantics
$\sem[C]{\phi}$. The following example shows a concrete case. 

\begin{expl}[Applying Bekič's bisection lemma]
  \label{exp:bekbislem}
   Consider the expression
  \[
    \phi = \nu x. \modalsym_1 ( x, \modalsym_2( x ), \nu y.
    \modalsym_3( y, \nu z . \modalsym_2( z ) ) )
  \]
  In order to transform it as per the procedure indicated, we first
  need to add a fixed point operator with a fresh variable to the
  first occurrence of $\modalsym_2$:
  \[
    \phi = \nu x. \modalsym_1 ( x, \nu w .\modalsym_2( x ), \nu y.
    \modalsym_3( y, \nu z . \modalsym_2( z ) ) )
  \]
Then we can form the equation system for the variables $x, w, y, z$
\begin{align*}
  x &= \overline\modalsym_1( x, w, y, z ) = \modalsym_1( x, w, y ) \\
  w &= \overline\modalsym_2( x, w, y, z ) = \modalsym_2( x ) \\
  y &= \overline\modalsym_3( x, w, y, z ) = \modalsym_3( y, z ) \\
  z &= \overline\modalsym_4( x, w, y, z ) = \modalsym_2( z )
\end{align*}
where we extend \modalset with additional operators
$\overline\modalsym_i$ having dummy arguments, defined as indicated.
The semantics of this equation system in a coalgebra $C=( X, \xi )$ is
defined as the greatest fixpoint $( A_0, A_1, A_2, A_3 )$ of the map
$h: \copow^n X \to \copow^n X$ defined by
\[
  h:
  \begin{pmatrix}
    X_1 \\ X_2 \\ X_3 \\ X_4
  \end{pmatrix}
  \mapsto
  \begin{pmatrix}
    \xi^{-1}[\sem[X]{\modalsym_1}( X_1, X_2, X_3 )] \\
    \xi^{-1}[\sem[X]{\modalsym_2}( X_1 )] \\
    \xi^{-1}[\sem[X]{\modalsym_1}( X_2, X_4 )] \\
    \xi^{-1}[\sem[X]{\modalsym_2}( X_4 )]
  \end{pmatrix}.
\]
The semantics of $\phi$ in $C$ is then $\sem[C]{\phi} = A_0$.
\end{expl}
\noindent The following two results together establish a Kleene-type
correspondence for the generic expressions of
\autoref{def:generic-expressions}.  
\begin{thm}
  \label{thm:generic-expression-describes-exactly-one-bisimulation-class}
  Every expression $\phi \in \gcexpr$ describes exactly one
  behavioural equivalence class, which is moreover realized in a
  finite coalgebra. Explicitly: there exists a state $x$ in a finite
  coalgebra such that for every coalgebra $C$, $\sem[C]{\phi}$
  contains precisely the states of $C$ that are behaviourally
  equivalent to~$x$.
\end{thm}
\begin{proof}[Sketch]
  By \autoref{lem:invariance}, it suffices to show that any two states
  (w.l.o.g.\ in the same coalgebra, using coproducts) satisfying
  $\phi$ are bisimilar. Since~$\phi$ can transformed into a
  system~\eqref{eq:system} of flat equations, this follows by
  \autoref{lem:gfp-is-singleton}. Realization in a finite coalgebra
  follows from the finite model property of the coalgebraic
  $\mu$-calculus~\cite{Cirstea11}, and alternatively is shown by
  constructing a model from the variables in a flat equation system.
\end{proof}

\begin{thm}
  \label{thm:kleene-theorem-for-generic-expressions}
  Let $C=(X, \xi)$ be a finite $T$-coalgebra. For every $x \in X$, there
  exists an expression $\phi \in \gcexpr$ such that
  $x \in \sem[C]{\phi}$.
\end{thm}
\begin{proof}
  Let $X = \{x_1, \dots, x_{k}\}$ and w.l.o.g.\@ $x = x_1$.  Since
  \pliftset is strongly expressive, for every $x_i \in X$ there is a
  modality $L_i$, w.l.o.g.\ with arity $k$ and prescribed argument
  ordering, such that
  \[
  \{ \xi(x_i) \} = \sem[X]{L_i}(\{x_{1}\}, \ldots, \{x_{k}\}).
  \]
  That is, the $\{x_i\}$ solve the system
  $(x_i=L_i(x_1,\dots,x_k))_{i=1,\dots,k}$ of flat fixed point
  equations, so for the greatest fixed point $(A_1,\dots,A_k)$ of the
  system, we have $x_i\in A_i$ for every~$i$, in particular
  $x=x_1\in A_1$. It now just remains to convert the equation system
  into an equivalent single expression in the standard
  manner~\cite{BradfieldStirling01} (incurring exponential blow-up);
  then $x\in\sem[C]{\phi}$ as desired.
\end{proof}
\begin{corollary}
  Every expression
  denotes a behavioural equivalence class of a state in a finite
  coalgebra, and conversely every such class is denoted by some
  expression.
\end{corollary}
\begin{expl}
  \begin{sparenumerate}
  \item For the functor $TX = 2 \times X^A$ for $A = \{a,b\}$ consider
    the coalgebra with carrier $X = \{x_1, x_2\}$ and with coalgebra
    structure $\xi: X \to 2 \times X^A$ with
    $\xi(x_0) = (1, (a \mapsto x_0, b \mapsto x_1))$ and
    $\xi(x_1) = (0, (a \mapsto x_1, b \mapsto x_0))$. Then we clearly
      have $\{\xi(x_1)\} = \plift^1(\{x_1\},\{x_2\})$ and
    $\{\xi(x_2)\} = \plift^1(\{x_2\},\{x_1\})$. Using the syntax
    of~\autoref{ex:expr}.1 and following the proof
    of~\autoref{thm:kleene-theorem-for-generic-expressions}, we obtain
    the following expression for the behavioural equivalence
    class (i.e.~formal language) for $x_1$: 
    \[
      \nu x_1.\modab 1 {x_1} {\nu x_2. \modab 0 {x_2} {x_1}}.
    \]
    Note that this is the same expression (modulo
    $\alpha$-equivalence) as the third expression
    from~\autoref{ex:expr}.1.

  \item For the functor $\fpow(A \times -)$ and $A = \{a,b\}$ the coalgebra $C=(\{x,y,z,w\},
    \xi)$ depicted in
    \autoref{ex:expr}.2 satisfies the following equations:
      \[
        \{\xi(x)\} = \plift_C^{(a)}(\{y\}), \;
        \{\xi(y)\} = \plift_C^{(a,b,c)}(\{x,w,z\}), \;
        \{\xi(w)\} = \plift_C^{()}(), \;
        \{\xi(w)\} = \plift_C^{()}()
      \]
      By \autoref{thm:kleene-theorem-for-generic-expressions} $\{\{x\},\{y\},\{z\},\{w\}\}$
      solves the following system, reusing the same variable names,
      \[
        x = [a](y), \quad
        y = [a,b,a](x,w,z),\quad
        w = [()], \quad
        z = [()]
      \]
      which can be transformed as demonstrated in \autoref{exp:bekbislem} to the expression
      given in \autoref{ex:expr}.2, describing the behaviour of the state $x$.
    \item For the functor $T = \dist$ consider the expression from \autoref{ex:expr}.3:
      $\nu x. [\nicefrac{2}{3}, \nicefrac{1}{3}](x, \nu y.
      [\nicefrac{1}{6},\nicefrac{1}{3},\nicefrac{1}{2}](x,y, \nu z.
      [\nicefrac{1}{4},\nicefrac{3}{4}](x,z)))$, which transforms to the system
      \[
        x = [\nicefrac{2}{3}, \nicefrac{1}{3}](x, y), \quad
        y = [\nicefrac{1}{6},\nicefrac{1}{3},\nicefrac{1}{2}](x,y, z), \quad
        z = [\nicefrac{1}{4},\nicefrac{3}{4}](x,z).
      \]
      By \autoref{thm:generic-expression-describes-exactly-one-bisimulation-class} we can
      construct a coalgebra $C =(\{x,y,z\}, \xi)$ defined by:
      \begin{align*}
        \{\xi(x)\} &= \plift_X^{(\nicefrac{2}{3}, \nicefrac{1}{3})}(\{x\},\{y\})\\
        \{\xi(y)\} &=
        \plift_X^{(\nicefrac{1}{6},\nicefrac{1}{3},\nicefrac{1}{2})}(\{x\},\{y\},\{z\})\\
        \{\xi(z)\} &= \plift_X^{(\nicefrac{1}{4},\nicefrac{3}{4})}(\{x\},\{z\})
      \end{align*}
      which is exactly the coalgebra depicted in \autoref{ex:expr}.3 where $x$ is in the
      behavioural equivalence class of the above expression.
  \end{sparenumerate}
\end{expl}
An alternative approach to defining the semantics of expressions is to
construct a $T$-coalgebra structure on the set \gcexpr of closed and
guarded expressions, similarly as in the work of Silva et
al.~\cite{SilvaBR10} and also Myers~\cite{M13}. In
\autoref{thm:coalgebraic-and-modallogic-semantics-coincide} below we
show that this new semantics coincides with the previous one.
\begin{defn}
  We define a $T$-coalgebra
$\gcemor : \gcexpr \rightarrow T \gcexpr$ inductively by
\begin{align}
  \gcemor( \modalsym ( \phi_1, \dots, \phi_n ) )  &\in
    \sem{L} ( \{ \phi_1 \}, \dots, \{ \phi_n \} )
      \label{equ:coalgstruc1}\\
  \gcemor( \nu x.\phi ) & = \varepsilon( \phi[ \nu x.\phi/x
    ] ). \label{equ:coalgstruc2}
\end{align}
\end{defn}
This is actually a definition of $\gcemor$ because (a)
in~\eqref{equ:coalgstruc1}, $\sem{L}$ preserves singletons and thus
there is only one element in
$\sem{L} ( \{ \phi_1 \}, \dots, \{ \phi_n \} )$, and (b)~for the
inductive part \eqref{equ:coalgstruc2}, one can use the number of
top-level fixed point operators as a termination measure, which
decreases in each step because the fixed points are
guarded.

Now recall that a coalgebra $\xi: X\to TX$ is \emph{locally finite} if
every $x\in X$ is contained in a finite subcoalgebra of $\xi$. Locally
finite coalgebras are precisely the (directed) unions of finite
coalgebras (see~\cite{Milius10}). Thus, it follows from
\autoref{thm:kleene-theorem-for-generic-expressions} that for any
$x\in X$ in a locally finite coalgebra $\xi: X\to TX$, there exists a
$\phi \in \gcexpr$ with $x\in \sem[X]{\phi}$.

Moreoever, $\gcexpr$ is obviously not finite; however, arguing via
finiteness of the Fischer-Ladner closure~\cite{Kozen83} we obtain
\begin{prop}\label{prop:locally-finite}
  The $T$-coalgebra $(\gcexpr,\gcemor)$ is locally finite.
\end{prop}

\noindent The following theorem says that $(\gcexpr,\gcemor)$ serves
as a canonical model of the expression language:

\begin{thm}
  \label{thm:coalgebraic-and-modallogic-semantics-coincide}
  For every closed and guarded expression $\phi\in\gcexpr$ and every
  state~$x$ in a $T$-coalgebra~$C$, $x\in\sem[C]{\phi}$ iff $x$ is
  behaviourally equivalent to $\phi$ as a state in
  $(\gcexpr,\gcemor)$.
\end{thm}
\noindent In particular, the above implies that
\begin{equation}
  \label{eq:truth-lem}
  \phi\in\sem[\gcexpr]{\phi}\qquad\text{for all $\phi\in\gcexpr$,}
\end{equation}
essentially a truth lemma for $\gcexpr$. For the proof of
Theorem~\ref{thm:coalgebraic-and-modallogic-semantics-coincide}, we note:
\begin{lem}\label{lem:alpha-eq}
  $\alpha$-Equivalent expressions are behaviourally equivalent as
  states in $(\gcexpr,\gcemor)$.
\end{lem}

\begin{proof}[Theorem~\ref{thm:coalgebraic-and-modallogic-semantics-coincide}, sketch]
  It suffices to prove~\eqref{eq:truth-lem}: The `if' direction of the
  claim then follows from invariance of $\phi$ under behavioural
  equivalence (\autoref{lem:invariance}), and `only if' is by
  Theorem~\ref{thm:generic-expression-describes-exactly-one-bisimulation-class}.
  We generalize~\eqref{eq:truth-lem} to expressions~$\phi$ with free
  variables: Whenever $\sigma$ is a substitution of the free variables
  of~$\phi$ and $\val$ a valuation such that $\sigma(v)\in\val(v)$ for
  every free variable~$v$ of $\phi$, then
  \begin{equation*}
    \phi\sigma\in\semv{\gcexpr}{\val}{\phi}.
  \end{equation*}
  We proceed by induction on $\phi$, using Lemma~\ref{lem:alpha-eq} in
  the fixpoint case.
\end{proof}

\begin{rem}
  To give a concrete example use of the connection between expression
  languages and modal fixed point logics afforded by the above
  results, we note that we now obtain an alternative handle on
  equivalence of expressions that complements the standard approach
  via partition refinement: Expressions $\phi,\psi$ are equivalent iff
  some state described by $\phi$ (obtained, e.g., via the one of the model
  constructions in
  Theorems~\ref{thm:generic-expression-describes-exactly-one-bisimulation-class}
  and~\ref{thm:coalgebraic-and-modallogic-semantics-coincide})
  satisfies~$\psi$. Note that the latter is fairly easy to check as
  long as the modalities are computationally tractable, since $\psi$
  otherwise involves only greatest fixed points. This approach is
  similar to reasoning algorithms in the lightweight description logic
  $\mathcal{EL}$~\cite{BaaderBrandtLutz05}, where checking validity of
  $\phi\to\psi$ is reduced to model checking $\psi$ in a minimal model
  of~$\phi$; we leave a more detailed analysis to future work.
\end{rem}

\section{Conclusion and Further Work}

We have defined a generic expression language for behaviours of finite
set coalgebras based on predicate liftings, specifically on a strongly
expressive set of singleton-preserving predicate liftings. There are
mutual conversions between such sets of predicate liftings and functor
presentations, one direction being via the Moss liftings introduced by
Marti and Venema~\cite{MartiV15}; we have however demonstrated that
these fail to be mutually inverse in one direction, i.e.\ in general
not all singleton-preserving predicate liftings are Moss liftings. Our
language is presumably equivalent to the set-based instance of Myer's
expression language~\cite{M13}; our alternative presentation is aimed
primarily at showing that expression languages embed naturally into
the coalgebraic $\mu$-calculus, generalizing well-known results on the
relational
$\mu$-calculus~\cite{GrafSifakis86,GodskesenEA87,SteffenIngolfsdottir94,AcetoEA07}. The
benefit of this insight is to tighten the connection between
expression languages and specification logics, e.g.\ it allows for
combining model checking, equivalence checking, and reasoning within a
single formalism.  On a more technical note, we show, e.g., that one
can provide an alternative semantics of expressions by defining a
coalgebra structure on expressions, an approach pioneered by Silva et
al.~\cite{SilvaBR10} and used also by Myers~\cite{M13}; in the light
of the expressions/logic correspondence, this construction is now seen
as a canonical model construction for a fragment of the coalgebraic
$\mu$-calculus, and the core part of the proof that the two semantics
agree becomes just a truth lemma.

An important point for further work is to extend the current setup
from the base category $\Set$ to algebraic categories (such as join
semi-lattices or positive convex algebras) in order to generalize our
results to expression calculi involving convenient additional
operations (reflecting the ambient algebraic theory) such as
addition. A closely related point is the connection with coalgebraic
determinization~\cite{SilvaEA13}; it should be interesting to see
whether our ideas can lead to expression calculi for coarser system
equivalences than bisimilarity, such as trace equivalence for
transition systems or distribution bisilimarity for Segala systems.
Such a generalization might be based on our recent approach
to coalgebraic trace semantics via graded monads~\cite{MiliusEA15}.

%
%
\iffull
\clearpage
\appendix

\section{Omitted Details and Proofs}

\subsection*{Proof of \autoref{lem:pres}}

  \emph{\ref{item:op-from-pl}.:} First note that~\eqref{eq:op-from-pl}
  is really a definition of $\cfpre_\plift$ because $\plift$ is
  preserves singletons. Moreover, $\cfpre_\plift$ is natural because
  for $f:X\to Y$ and $x_i\in X$,
  \begin{align*}
    \cfpre_{\plift, X}( x_1,\ldots,x_n ) & \in \plift_X( \{ x_1
    \},\ldots,\{ x_n \} ) & \tag*{(by definition)}\\
      & \subseteq \plift_X( f^{-1}[ \{ f(x_1) \}
      ],\ldots,f^{-1}[ \{ f(x_n) \} ] ) & \tag*{(\plift monotone)} \\
      & = (Tf)^{-1}[ \plift_Y( \{ f(x_1) \},\ldots,\{ f(x_n)
      \} ) ] & \tag*{(naturality of \plift)}\\
      & = (Tf)^{-1}[ \{ \cfpre_{\plift, Y}( f(x_1),\ldots,f(x_n) )
      \}] & \tag*{(by definition)}
    \\
    \Rightarrow\ &Tf ( \cfpre_{\plift, X} ( x_1,\ldots,x_n )) =
    \cfpre_{\plift, Y}( f(x_1),\ldots,f(x_n) )
  \end{align*}

  \emph{\ref{item:pres-from-pls}.:} By the previous item, we obtain a
  natural transformation
  $\fpre = [ \cfpre_\plift ]_{\plift\in\pliftset}$, and
  strong expressivity implies that $\fpre$ is componentwise
  surjective.
\qed

\subsection*{Proof of \autoref{prop:laxdescrpres}}
  Recall that $\plift = \nabla^L\circ\cfpre \copow$ by definition. Let
  $s_X: X \rightarrow \copow X$ be the function $s_X( x ) = \{ x \}$. Then we have
    \begin{align*}
      & \plift_X( \{x_1\},\ldots,\{x_{n}\} ) \\
      & = \{ t\in TX \mid ( t, \cfpre_{\copow X}( \{ x_1
        \},\ldots,\{ x_{n} \} ) ) \in L(\in_X) \} \\
      & = \{ t\in TX \mid ( t, \cfpre_{\copow X}( s_X( x_1
        ),\dots,s_X( x_{n} ) ) ) \in L(\in_X) \} \\
      & = \{ t\in TX \mid ( t, Ts_X \circ \cfpre_X( x_1,\dots,x_{n}
        ) ) \in L(\in_X) \} \\
      & = \{ t\in TX\mid( t,\cfpre_X( x_1,\dots,x_{n} )
        )\in L(\in_X);\irel{( Ts_X )} \} \\
      & = \{ t\in TX\mid( t,\cfpre_X( x_1,\dots,x_{n} )
        )\in L(\in_X ; \irel{s_X}) \} \tag*{(Proposition~\ref{prop:laxprop})}\\
      & = \{ t\in TX\mid( t,\cfpre_X( x_1,\dots,x_{n} )
        )\in L(\Delta_X) \} \tag*{$(s_X ; \mathord{\ni_X} = \Delta_X)$} \\
      & = \{ \cfpre_X( x_1,\dots,x_{n} ) \}
        \tag*{\qed}
    \end{align*}

\subsection*{Proof of \autoref{prop:moss-strong-expr}}
  Let $t\in TX$; we need to show that
  $\{ t \} = \plift( \{ x_1 \},\dots,\{ x_n \} )$ for some Moss
  lifting $\plift$ and $x_1, \ldots, x_n\in X$. Since
  $\fpre: T_\Sigma \to T$ from our functor presentation of $T$ has
  surjective components, there exists some $\cfpre/n \in \Sigma$ such
  that $t = \fpre_X(x_1, \ldots, x_n) = \cfpre_X( x_1,\dots,x_n )$ for
  some $x_1, \ldots, x_n\in X$.
  Because $L$ is a lax extension, we have
  $Ts_X \subseteq Ls_X \subseteq L(\in_X)$ for the function
  $s_X : x\mapsto \{ x \}$ and thus
  \[
    (t, Ts_X(t)) = ( t, \cfpre_{\copow X}( \{ x_1 \},\dots,\{ x_n \} )
    ) \in L(\in_X)
  \]
  since
  $Ts_X(t)
  = Ts_X(\cfpre_X(x_1, \ldots, x_n))
  = \cfpre_{\copow X} (s_X(x_1), \ldots, s_X(x_n))$.
  Thus $t \in \plift( \{ x_1 \},\dots,\{ x_n \} )$
  (see~\autoref{def:Mosslift}), and since Moss liftings 
  preserve singletons, we conclude
  $\{ t \} = \plift( \{ x_1 \},\dots,\{ x_n \} )$ as desired. \qed

\subsection*{Details for \autoref{rem:moss-unary}}

We show that $\plift = \plift^{(\cfpre^\plift)}$ for monotone
singleton-preserving \emph{unary} predicate liftings~$\plift$. Let
$\plift/1$ and $\plift'/1$ be predicate liftings that agree on
singletons. Now for any $X$ and $A \subseteq X$ consider the
characteristic map $\chi_A: X \to 2 = \{\bot, \top\}$, which satisfies
$A = \chi_A^{-1}(\top)$. Then using naturality and the fact that
$\plift_2(\{\top\}) = \plift_2'(\{\top\})$, we have
  \begin{align*}
    \plift_X(A)
    & = \plift_X(\chi_A^{-1}(\{\top\}))
      = (T\chi_A)^{-1}(\plift_2(\{\top\}))\\
    & = (T\chi_A)^{-1}(\plift_2'(\{\top\}))
      = \plift_X'(\chi_A^{-1}(\{\top\}))
      = \plift_X'(A).
  \end{align*}
  Now observe that $\plift$ and $\plift^{(\cfpre^\plift)}$ agree on
  singletons, thus they are equal as desired. 

\subsection*{Proof of \autoref{lem:fp-substitution}}
  First we prove that for any expression $\phi\in\expr$ that does not contain free
  variables that are bound in another expression $\psi\in\expr$ the following holds:
  \begin{gather}
    \sem\psi^{\val[ z\mapsto\sem\phi^\val ]} = \sem{\psi[ \phi/z ]}^\val
    \label{eq:substitutionlemma}
  \end{gather}
  Induction on $\psi$. If $\psi = x \neq z$ then
  $\sem x^{\val[ z\mapsto\sem\phi^\val ]} = \sem x^\val = \sem{x[ \phi/z
  ]}^\val$; if $\psi = z$ then $\sem z^{\val[ z\mapsto\sem\phi^\val ]} =
  \sem\phi^\val = \sem{z[ \phi/z ]}^\val$; for $\psi = \modalsym(
  \phi_1,\ldots,\phi_n )$:
  \begin{gather*}
    \sem{\modalsym( \psi_1,\ldots,\psi_n )}^{\val[ z\mapsto\sem\phi^\val
    ]} = \xi^{-1}[ \sem{L}( \sem{\psi_1}^{\val[ z\mapsto\sem\phi^\val
    ]},\ldots,\sem{\psi_n}^{\val[ z\mapsto\sem\phi^\val ]} ) ]\\
    \overset{I.H.}= \xi^{-1}\big[ \sem{L}( \sem{\psi_1[ \phi/z
    ]},\ldots,\sem{\psi_n[ \phi/z ]} ) \big] = \sem{\modalsym(
    \psi_1[ \phi/z ],\ldots,\psi_n[ \phi/z ] )}\\
    = \sem{\modalsym( \psi_1,\ldots,\psi_n )[ \phi/z ]};
  \end{gather*}
  if $\psi = \nu z.\psi_1$ then $\sem{\nu z.\psi_1}^{\val[
    z\mapsto\sem\phi^\val ]} = \sem{\nu z.\psi_1}^\val = \sem{(\nu z.\psi)[ \phi/z
  ]}^\val$; if $\psi = \nu x.\psi_1$ with $x \neq z$ then:
  \begin{gather*}
    \sem{\nu x.\psi_1}^{\val[ z\mapsto\sem\phi^\val ]}
    = \nu Y.\sem{\psi_1}^{\val[ z\mapsto\sem\phi^\val ][ x\mapsto Y ]}
    = \nu Y.\sem{\psi_1}^{\val[ x\mapsto Y ][ z\mapsto\sem\phi^\val ]}\\
    \overset{I.H.}= \nu Y.\sem{\psi_1[ \phi/z ]}^{\val[ x\mapsto Y ]}
    = \sem{\nu x.( \psi_1[ \phi/z ] )}^{\val} = \sem{( \nu x.\psi_1
    )[ \phi/z ]}^\val
  \end{gather*}
  where the last equation holds because of the assumption that the free variables of $\phi$
  are not bound in $\psi$.
  
  \noindent It follows that
  \[
    \sem{\nu z.\phi} = \nu Y.\sem\phi^{[z\mapsto Y]} = \sem\phi^{[z\mapsto
    \sem{\nu z.\phi}]}
    \overset{\eqref{eq:substitutionlemma}}= \sem{\phi[ \nu z.\phi / z ]}.
  \]

\subsection*{Proof of \autoref{lem:predicate-liftings-preserve-equivalence-classes}}

Denote by $q: X \epi X/S$ the canonical quotient map, let
$C_1,\ldots,C_m \subseteq X$ be $S$-closed, and let $\plift/m$ be monotone
and preserve singletons. Then $C_i = q^{-1} [ q[C_i]]$ for
all~$i$. Therefore, we have for every $t\in \copow TX$:
\begin{gather}
    t \in \plift_{X}(C_1,\dots C_m) =
      \plift_{X}(q^{-1}[q[C_1]],\dots,
      q^{-1}[ q [ C_m ] ]) \nonumber \\
    \Rightarrow
    t \in ( Tq )^{-1} [\plift_{X/S}( q[ C_1 ],\dots,q[ C_m
      ] )] \label{eq:gat} \\
    \Rightarrow Tq(t) \in \plift_{X/S}(q[C_1],\dots, q[ C_m
      ]).\nonumber
\end{gather}

\begin{sparenumerate}
\item If there is some $t \in \plift_{X} (A_1,\dots A_n) \cap \plift_{X}'
(B_1,\dots B_{n'})$, then we show the inclusion. For $s\in \plift_{X} (A_1,\dots
A_n)$, \eqref{eq:gat} provides
\[
  \{Tq(t), Tq(s)\} \subseteq \plift_{X/S}(q[A_1],\ldots,q[A_n]).
\]
Every $A_i\subseteq X$ is an $S$-equivalence class or empty, so $q[A_i]$
is at most a singleton. By monotonicity and singleton preservation of $\plift$,
the right-hand side is at most a singleton, and thus $Tq(t) = Tq(s)$. Applying
$\eqref{eq:gat}$ to $t$ and $B_1,\ldots,B_{n'}$, we have
$Tq(s) = Tq(t) \in \plift'(q[B_1],\ldots,q[B_{n'}])$, and consequently $s\in
\plift'(B_1,\ldots,B_{n'})$, again by \eqref{eq:gat}.

\item
  In the case where the sets are not disjoint, the equality is
  obtained by proving both inclusions using point~(1).  
  \qed
\end{sparenumerate}

\subsection*{Proof of \autoref{lem:gfp-is-singleton}}

  Let $S \subseteq X \times X$ be the relation
  \[
    S = \{ ( x_1,x_2 )\mid \exists A_i\,.\,x_1\in A_i\wedge x_2\in A_i
    \}\cup\Delta_X.
  \]
  (1)~We prove that $S$ is an equivalence relation.  Let~$\overline S$
  be the equivalence relation generated by~$S$, and let
  $\overline{A_i}$ be the $S$-closure of $A_i$. Then the
  $\overline{A_i}$ are also $\overline S$-closed, and  we have
  \[
    \overline S = \{ ( x_1,x_2 )\mid\exists
    \overline{A_i}\,.\,x_1\in\overline{A_i}\wedge x_2\in\overline{A_i}
    \}\cup\Delta_X.
  \] 
  It follows that each $\overline{A_i}$ is either empty (if
  $A_i=\emptyset$) or else an $\overline S$-equivalence
  class.

  Now we show that $( \overline{A_1},\dots,\overline{A_k} )$ is a
  post-fixed point of $h$: Because $(A_1,\dots, A_n)$ is a fixed point
  and the $\plift_i$ are monotone, we have
  \begin{equation}\label{eq:fp}
    A_j = \xi^{-1} [ \plift_{j,X} (A_1,\dots,
    A_k)] \subseteq \xi^{-1} [\plift_{j,X} (\overline{A_1},\dots,
    \overline{A_k} )]\quad\text{for all $j$}.
  \end{equation}
  We will show that the post-fixed point condition
  \begin{equation}\label{eq:postfix}
    \overline{A_j} \subseteq \xi^{-1}[ \plift_{j,X}
    (\overline{A_1},\dots, \overline{A_k} )]
  \end{equation}
  holds for all $j$. To see this, it suffices by~\eqref{eq:fp} and the
  definition of $\overline{A_j}$ to show that whenever
  $A_l\cap A_p\neq\emptyset$, then
  \begin{equation}\label{eq:lambda-stable}
    \plift_{l,X} (\overline{A_1},\dots,
    \overline{A_k} ) = \plift_{p,X}
    (\overline{A_1},\dots, \overline{A_k}).
  \end{equation}
  So let $x\in A_l\cap A_p$. Then by~\eqref{eq:fp},
  $ \xi( x ) \in \plift_{l,X} (\overline{A_1},\dots,
  \overline{A_k} ) \cap \plift_{p,X}
  (\overline{A_1},\dots, \overline{A_k})$,
  and therefore~\eqref{eq:lambda-stable} follows by
  \autoref{lem:predicate-liftings-preserve-equivalence-classes}.2.

  Having shown that $(\overline{A_1},\dots,\overline{A_k})$ is a
  post-fixed point of $h$, we obtain that
  $\overline{A_i} \subseteq A_i$ for all $i$, which implies that
  $S = \overline S$ is an equivalence relation. Hence every $A_i$ is
  either empty or an equivalence class of $S$, whence
  $A_i \cap A_j = \emptyset$ or $A_i = A_j$ holds for all $i, j$ as
  claimed in the second part of the lemma.

  \medskip\noindent (2)~Now we prove that $S$ is a
  $\pliftset$-bisimulation, i.e.\@ $x \mathbin{S} y$ and
  $\xi(x) \in \plift_X (B_1,\dots,B_n)$ implies
  $\xi(y) \in \plift_X ( S[B_1], \dots, S[ B_n ])$ for every
  $\plift/n\in\pliftset$ and $B_j \subseteq X$. By
  Theorem~\ref{thm:Lambda}, this implies the first claim of the lemma,
  since $x,y\in A_j$ implies $x\,S\,y$ by construction of~$S$.

  By point~(1), $S$ is an equivalence relation, so
  $S[B_j] = \overline{B_j} \supseteq B_j$ is the $S$-closure of $B_j$;
  and since $\plift$ is monotone,
  $\xi(x) \in \plift_X (\overline{B_1}, \dots,
  \overline{B_n})$. Further, $x\mathbin{S}y$ implies that either $x=y$, in which
  case there is nothing to prove, or there exists~$j$ such that
  $x,y\in A_j$ and therefore
  \[
    \xi(x), \xi(y) \in \plift_{j,X}(A_1, \dots, A_k).
  \]
  Because $\xi( x )$ also lies in $\plift_X (\overline{B_1}, \dots,
  \overline{B_n}) $, we obtain
  \[
    \plift_{j,X}(A_1, \dots, A_k)
    \subseteq \plift_X (\overline{B_1}, \dots, \overline{B_n})
  \]
  by \autoref{lem:predicate-liftings-preserve-equivalence-classes},
  and thus $\xi(y) \in \plift_X (S[B_1], \dots, S[ B_n ])$. \qed

\subsection*{Proof of \autoref{lem:bekic}}
  By Tarski, all fixed points exist. $(x_0,y_0)$ is a fixed point of
  $\langle f,g  \rangle: X\times Y\to X\times Y$, because
  $y_0 = g(x_0,y_0)$ and thus also
  \[
    x_0 = f(x_0, \nu y.g(x_0,y)) = f(x_0,g(x_0,y_0)) = f(x_0,y_0).
  \]
  For any other fixed point $(x',y')$ of $\langle f,g \rangle$, we have
  \[
    x' = f(x',y') = f(x', g(x',y'))
    \le f(x', \nu y.g(x',y)),
  \]
  so $x' \le x_0$ and furthermore
  \(
    y' = g(x',y') \le g(x_0,y')
    \le \nu y.g(x_0,y)
  \).\qed

\subsection*{Proof of \autoref{thm:generic-expression-describes-exactly-one-bisimulation-class}}
  We first show that $\phi$ denotes a single behavioural equivalence
  class, i.e.\ (i)~we have invariance under behavioural equivalence
  (\autoref{lem:invariance}) and (ii) two states satisfying $\phi$ are
  bisimilar. In order to prove~(ii) we may w.l.o.g.~assume that the
  two states live in the same coalgebra; indeed, given two coalgebras
  $C$ and $D$ and states $x$ in $C$ and $y$ in $D$ satisfyfing $\phi$,
  then $\inl(x), \inr(y)$ in $C + D$ are behaviourally equivalent to
  $x$ and $y$, respectively, thus both satisfy $\phi$ by
  \autoref{lem:invariance}.

  Now let $C=(X,\xi)$ be a $T$-coalgebra; we need to
  show that any two elements of $\sem[C]{\phi}$ are behaviourally
  equivalent. As noted above, we can transform~$\phi$ into a
  system~\eqref{eq:system} of flat equations, and $\sem[C]{\phi}$ is
  the first component of the greatest fixed point of the corresponding
  map of the form~\eqref{eq:flat-gfp}. The
  claim thus follows by \autoref{lem:gfp-is-singleton}.
  
  It remains to show that there exists a finite coalgebra $C$ such
  that $\sem[C]{\phi}\neq\emptyset$. This is immediate from the finite
  model property of the coalgebraic $\mu$-calculus~\cite{Cirstea11};
  alternatively, avoiding such overkill, it is seen as follows: We
  define $C=(X,\xi)$ on the set~$X=\{x_1,\dots,x_k\}$ of the variables
  in the above flat equation system by
  $\xi(x_i)\in\sem[X]{L_i}(\{x_1\},\dots,\{x_k\})$ where
  $\phi_i=L_i(x_1,\dots,x_k)$. Then by construction,
  $(\{x_1\},\dots,\{x_k\})$ is a fixed point of the equation system,
  hence contained in the greatest fixed point, which proves the desired
  non-emptiness of the greatest fixed point.
%
\subsection*{Proof of \autoref{prop:locally-finite}}
  Recall that the Fischer-Ladner closure of a $\mu$-calculus formula
  is standardly defined as the closure under subformulas, negation,
  and fixed point unfolding. In the absence of negation, we adapt the
  definition to include only closure under subformulas and fixed point
  unfolding. We then, of course, inherit the standard result that the
  Fischer-Ladner closure is finite~\cite{Kozen83}. Now by construction
  of $\gcemor$, the subcoalgebra of $(\gcexpr,\gcemor)$ generated by
  $\phi \in \gcexpr$ contains only states from the
  Fisher-Ladner-closure of $\phi$, hence is finite. \qed

\subsection*{Proof of \autoref{thm:coalgebraic-and-modallogic-semantics-coincide}}

  It suffices to prove the truth lemma~\eqref{eq:truth-lem}: The `if'
  direction of the claim then follows from invariance of $\phi$ under
  behavioural equivalence (\autoref{lem:invariance}), and `only if' is
  by
  Theorem~\ref{thm:generic-expression-describes-exactly-one-bisimulation-class}.

  We strengthen~\eqref{eq:truth-lem} to a claim on expressions~$\phi$
  possibly having free variables: Whenever $\sigma$ is a substitution
  of the free variables of~$\phi$ and $\val$ a valuation such that
  $\sigma(v)\in\val(v)$ for every free variable~$v$ of $\phi$, then
  \begin{equation}
    \label{eq:truth-open}
    \phi\sigma\in\semv{\gcexpr}{\val}{\phi}.
  \end{equation}
  We prove~\eqref{eq:truth-open} by induction over $\phi$. The case
  for fixed point variables is just by the assumption on $\sigma$ and
  $\val$. For the modal case, we calculate as follows:
  \begin{align*}
    L(\phi_1,\dots,\phi_n)\sigma & = L(\phi_1\sigma,\dots,\phi_n\sigma)\\
    & \in\gcemor^{-1}[\sem{L}(\{\phi_1\sigma_1\},\dots,\{\phi_n\sigma\})
      & \tag{by definition}\\
    & \subseteq\gcemor^{-1}[\sem{L}(\semv{\gcexpr}{\val}{\phi_1},\dots,
      \semv{\gcexpr}{\val}{\phi_n}\})
      & \tag{induction, monotonicity}\\
    & = \semv{\gcexpr}{\val}{L(\phi_1,\dots,\phi_n)}. &\tag{semantics}
  \end{align*}
  Finally, for the fixed point case $\nu x.\,\phi$, first note that by
  Lemma~\ref{lem:alpha-eq}, we can assume that $x$ does not occur as a
  free variable in $\sigma(v)$ for any free variable~$v$
  of~$\nu x.\,\phi$. Moreover, since $x$ is not free in
  $\nu x.\,\phi$, $\sigma$ does not touch~$x$. Thus,
  $(\nu x.\,\phi)\sigma=\nu x.\,(\phi\sigma)$. By construction of
  $\gcemor$, $\nu x.\,(\phi\sigma)$ is, as a state of~$\gcexpr$,
  behaviourally equivalent to $\phi\sigma[\nu x.\,\phi\sigma/x]$,
  which by the inductive hypothesis is contained in
  $\semv{\gcexpr}{\kappa'}{\phi}$ where $\kappa'$ arises from $\kappa$
  by assigning to~$x$ the value
  $\{\nu x.\,\phi\sigma\}=\{(\nu x.\,\phi)\sigma\}$. This shows that
  $\{(\nu x.\,\phi)\sigma\}$ is a post-fixed point of the map defining
  $\semv{\gcexpr}{\val}{\nu x.\,\phi}$, and hence contained in
  $\semv{\gcexpr}{\val}{\nu x.\,\phi}$, which proves the claim. \qed

\subsection*{Proof of \autoref{lem:alpha-eq}}

  Let $\pi_\alpha:\gcexpr\to\gcexpr/\alpha$ denote the quotient map of
  $\gcexpr$ modulo $\alpha$-equivalence. Paralleling the definition of
  $\gcemor$, we define a $T$-coalgebra structure $\gcemor_\alpha$ on
  $\gcexpr/\alpha$ by
  \begin{align*}
      \gcemor_\alpha( \pi_\alpha(\modalsym ( \phi_1, \dots, \phi_n ) ))  &\in
    \sem{L} ( \{ \pi_\alpha(\phi_1) \}, \dots, \{ \pi_\alpha(\phi_n) \} )\\
  \gcemor( \pi_\alpha(\nu x.\phi) ) & = \varepsilon( \pi_\alpha(\phi[ \nu x.\phi/x
    ]) ). 
  \end{align*}
  One sees in largely the same way as for $\gcemor$ that this is
  actually a definition, noting additionally that $\alpha$-equivalent
  transformations of a formula $L(\phi_1,\dots,\phi_n)$ necessarily
  happen in its arguments $\phi_i$ and that the number of top-level
  fixed point operators is invariant under $\alpha$-equivalence. We
  are done once we show that
  $\pi_\alpha:(\gcexpr,\gcemor)\to(\gcexpr/\alpha,\gcemor_\alpha)$ is
  a coalgebra morphism. We proceed by case distinction over the shape
  of states $\phi\in\gcexpr$:

  First assume that $\phi$ has the form
  $\phi=L(\phi_1,\dots,\phi_n)$. We have to show
  $T\pi_\alpha(\gcemor(L(\phi_1,\dots,\phi_n)))\in\sem{L}(\{\pi_\alpha(\phi_1)\},\dots,\{\pi_\alpha(\phi_n)\})$. By
  naturality, this is equivalent to
  $\gcemor(L(\phi_1,\dots,\phi_n))\in\sem{L}(\pi_\alpha^{-1}[\{\pi_\alpha(\phi_1)\}],\dots,\pi_\alpha^{-1}[\{\pi_\alpha(\phi_n)\}])$,
  which follows by monotonicity from the fact that
  $\gcemor(L(\phi_1,\dots,\phi_n))\in\sem{L}(\{\phi_1\},\dots,\{\phi_n\})$
  by definition.

  Second, assume that $\phi$ has the form $\phi=\nu x.\psi$. We
  proceed by induction on the number of top-level fixed point
  operators in $\phi$: We have
  \begin{align*}
T\pi_\alpha(\gcemor(\nu x.\psi)) & =T\pi_\alpha(\gcemor(\psi[\nu x.\psi/x])) 
                                   & \tag*{(by definition)}\\
    & = \gcemor_\alpha(\pi_\alpha(\psi[\nu x.\psi/x]))
      & \tag*{(induction)}\\
    & \gcemor_\alpha(\pi_\alpha(\nu x.\psi))
                                   & \tag*{(by definition)}.
  \end{align*}
  Since $\phi$ is closed, these are the only cases. \qed

\fi

\bibliographystyle{splncs04}
\bibliography{bibcmcs18}

\end{document}


%% file: defslncs.tex
\spnewtheorem{thm}[theorem]{Theorem}{\bfseries}{\itshape}
\spnewtheorem{cor}[theorem]{Corollary}{\bfseries}{\itshape}
\spnewtheorem{cnj}[theorem]{Conjecture}{\bfseries}{\itshape}
\spnewtheorem{lem}[theorem]{Lemma}{\bfseries}{\itshape}
\spnewtheorem{lemdefn}[theorem]{Lemma and Definition}{\bfseries}{\itshape}
\spnewtheorem{prop}[theorem]{Proposition}{\bfseries}{\itshape}
\spnewtheorem{defn}[theorem]{Definition}{\bfseries}{\upshape}
\spnewtheorem{rem}[theorem]{Remark}{\bfseries}{\upshape}
\spnewtheorem{notation}[theorem]{Notation}{\bfseries}{\upshape}
\spnewtheorem{expl}[theorem]{Example}{\bfseries}{\upshape}
\spnewtheorem{thmdefn}[theorem]{Theorem and Definition}{\bfseries}{\itshape}
\spnewtheorem{propdefn}[theorem]{Proposition and Definition}{\bfseries}{\itshape}
\spnewtheorem{assumption}[theorem]{Assumption}{\bfseries}{\upshape}
\spnewtheorem{algorithm}[theorem]{Algorithm}{\bfseries}{\upshape}

 \renewenvironment{corollary}{\begin{cor}}{\end{cor}}

 \renewenvironment{remark}{\begin{rem}}{\end{rem}}
 \renewenvironment{example}{\begin{expl}}{\end{expl}}